\documentclass[leqno]{article}

\usepackage[letterpaper]{geometry}

\usepackage[utf8]{inputenc}
\usepackage[T1]{fontenc}
\usepackage[english]{babel}

\usepackage{amsmath,amssymb,amsfonts,amsthm}
\usepackage{nicefrac}
\usepackage{csquotes}
\usepackage{xspace}
\usepackage{hyperref}
\usepackage{color}
\usepackage{graphicx}
\usepackage[labelfont=bf]{caption}
\usepackage{subcaption}
\usepackage[inline]{enumitem}
\usepackage{wrapfig}
\usepackage{mathtools}
\usepackage{varioref}

\newcommand{\ANDRE}[1]{\noindent\textcolor{red}{\textbf{André says:} #1}}
\newcommand{\MARVIN}[1]{\noindent\textcolor{green}{\textbf{Marvin says:} #1}}
\renewcommand{\ANDRE}[1]{}
\renewcommand{\MARVIN}[1]{}
\newcommand{\andre}[1]{\ANDRE{#1}}
\newcommand{\marvin}[1]{\MARVIN{#1}}

\newtheorem{thm}{Theorem}[section]
\newtheorem{lem}{Lemma}[section]

\newtheorem{prop}{Proposition}[section]
\newtheorem{cor}{Corollary}[section]
\newtheorem{defn}{Definition}[section]
\newtheorem{fact}{Fact}[section]

\newtheorem{claim}{Claim}[section]
\newtheorem{hypo}{Hypothesis}[section]

\newcommand{\norm}[1]{\left\lVert#1\right\rVert}

\newcommand{\Oh}{{\cal O}}
\newcommand{\tOh}{\tilde{\cal O}}

\newcommand{\eps}{\varepsilon}
\newcommand{\RR}{\mathbb{R}}

\newcommand{\x}{\mathrm{x}}
\newcommand{\y}{\mathrm{y}}
\newcommand{\rot}{\mathrm{rot}}
\newcommand{\trans}{\mathrm{trans}}
\newcommand{\scale}{\mathrm{scale}}

\newcommand{\rank}{\mathrm{rank}}
\newcommand{\st}{\mathrm{start}}
\newcommand{\en}{\mathrm{end}}

\newcommand{\average}{\textsc{Average}\xspace}

\usepackage{algorithm}
\usepackage[noend]{algpseudocode}

\usepackage{footnote}
\makesavenoteenv{algorithmic}
\makesavenoteenv{algorithm}

\usepackage{enumitem}
\setitemize{noitemsep,topsep=0pt,parsep=0pt,partopsep=0pt}

\usepackage{mdframed}

\def\cqedsymbol{\ifmmode$\lrcorner$\else{\unskip\nobreak\hfil
\penalty50\hskip1em\null\nobreak\hfil$\lrcorner$
\parfillskip=0pt\finalhyphendemerits=0\endgraf}\fi}

\begin{document}

\newcommand\relatedversion{}

\title{\Large Polygon Placement Revisited: (Degree of Freedom + 1)-SUM Hardness and an Improvement via Offline Dynamic Rectangle Union\relatedversion}
\author{Marvin Künnemann\thanks{ETH Z\"urich, Institute for Theoretical Studies. Research supported by Dr. Max R\"ossler, the Walter Haefner Foundation and the ETH Z\"urich Foundation.} \and
André Nusser\thanks{Max Planck Institute for Informatics and Saarbrücken Graduate School of Computer Science, Saarland Informatics Campus.}}

\date{}

\maketitle






\begin{abstract}
\small\baselineskip=9pt 
We revisit a classical problem in computational geometry: Determine the largest copy of a simple polygon~$P$ that can be placed into the simple polygon $Q$. Despite significant effort studying a number of settings, known algorithms require high polynomial running times, even for the interesting case when either $P$ or $Q$ have constant size. (Barequet and Har-Peled, 2001) give a conditional lower bound of $n^{2-o(1)}$ under the 3SUM conjecture when $P$ and $Q$ are (convex) polygons with $\Theta(n)$ vertices each. This leaves open whether we can establish (1) hardness beyond quadratic time and (2) any superlinear bound for constant-sized $P$ or $Q$.

In this paper, we affirmatively answer these questions under the higher-order $k$SUM conjecture, proving natural hardness results that increase with each degree of freedom (scaling, $x$-translation, $y$-translation, rotation):
	\begin{itemize}
		\item \emph{(scaling, $x$-translation:)} Finding the largest copy of $P$ that can be $x$-translated into $Q$ requires time $n^{2-o(1)}$ under the 3SUM conjecture, even for orthogonal (rectilinear) polygons $P,Q$ with $O(1)$ and $n$ vertices, respectively.
		\item \emph{(scaling, $x$-translation, $y$-translation:)} Finding the largest copy of $P$ that can be arbitrarily translated into $Q$ requires time $n^{2-o(1)}$ under the 4SUM conjecture, even for orthogonal polygons $P,Q$ with $O(1)$ and $n$ vertices, respectively. This establishes the same lower bound under the assumption that Subset Sum cannot be solved in time $O((2-\varepsilon)^{n/2})$ for any $\varepsilon > 0$.
		\item The above lower bounds are almost tight when one of the polygons is of constant size: Using an offline dynamic algorithm for maintaining the area of a union of rectangles due to Overmars and Yap, we obtain an $\tOh((pq)^{2.5})$-time algorithm for orthogonal polygons $P,Q$ with $p$ and $q$ vertices, respectively. This matches the lower bounds up to an $n^{1/2 + o(1)}$-factor when $P,Q$ have $O(1)$ and $n$ vertices.
		\item \emph{(scaling, $x$-translation, $y$-translation, rotation:)} Finally, finding the largest copy of $P$ that can be arbitrarily rotated and translated into $Q$ requires time $n^{3-o(1)}$ under the 5SUM conjecture.
	\end{itemize}
	As in our reductions, each degree of freedom determines one summand in a $k$SUM instance, these lower bounds appear likely to be best possible under $k$SUM. We are not aware of any other such natural $(\text{degree of freedom} + 1)$-SUM hardness for a geometric optimization problem. Finally, we prove an additional tight OV hardness of the translations-only case.
\end{abstract}


\section{Introduction} \label{sec:intro}

Consider the following basic question about two given simple polygons $P$ and $Q$: What is the largest similar copy of~$P$ that can be placed into $Q$ using translation and rotation? Besides its appeal as a natural geometric optimization problem, variants of this question occur in applications in manufacturing (e.g., cutting the largest possible version of a given shape in a given polygonal environment), motion planning (e.g., finding a \emph{high-clearance} path for an object, i.e., a path with largest-possible spacing between object and obstacles) and others. As an algorithmic optimization task involving several geometric degrees of freedom (scaling, translation in two dimensions, and rotation), it has led to a rich landscape of algorithmic works on a variety of settings, including (but not limited to)~\cite{Fortune85, LevenS87, ChewK93, SharirT94, AgarwalAS98, AgarwalAS99}. 

Among the various studied settings, we have restrictions of the degrees of freedom (such as a \emph{fixed-size} $P$ -- i.e., no scaling -- or no rotation), convexity restrictions ($P$ or $Q$ are convex), or restrictions of the shape (such as \emph{orthogonal} -- i.e., rectilinear -- polygons). In particular, the fixed-size setting (with or without rotation) has been heavily studied both for general $P$ and $Q$~\cite{Chazelle83, AvnaimB89}, for convex $P$ or convex $Q$~\cite{Chazelle83, AvnaimB89, SharirT94}, for orthogonal $P$ and $Q$~\cite{Barrera96algo}, or for orthogonally convex $P$ and $Q$~\cite{BakerFM86}. For obtaining the largest scaling when $P$ is convex, \cite{Fortune85} and \cite{LevenS87} study the case without rotation, while \cite{ChewK93, SharirT94, AgarwalAS98, AgarwalAS99} take rotation into account. Note that when $P$ is convex, then the existence of a placement of a  $\lambda$-scaled copy of $P$ into $Q$ is a monotone property in $\lambda$, and parametric search techniques can be used (see~\cite{SharirT94} for a thorough discussion). 
For an additional overview, see~\cite{HandbookDCGpolygons}, \cite{AgarwalSsurvey}.

Barely any of these results have near-linear running time, yet they stand for more than two decades: When $P$ is convex, the current state of the art for randomized algorithms is $\tOh( p^2 q^2)$ due to Agarwal, Amenta, and Sharir~\cite{AgarwalAS99} where $p$ and $q$ denote the number of vertices of $P$ and $Q$, respectively.\footnote{Throughout this paper $\tOh(T) = \Oh(T \mathrm{polylog}(T))$ is used to hide polylogarithmic factors.} When $Q$ is convex as well, this running time can be improved to $\tOh(pq^2)$~\cite{AgarwalAS98}. Note that these running times are not close to linear even for the natural special case of $P$ being a small pattern polygon of constant size $p=\Oh(1)$. Even the fixed-size case  (no scaling) for simple $P$ and $Q$ is challenging, with a current time bound of $\tOh(p^3 q^3)$. Are such high polynomial running times really necessary, and if so, why?

A common type of evidence would be to give corresponding lower bounds on the arrangement size encountered in their algorithmic solutions. Indeed, for simple $P$ and $Q$ it is known that the Minkowski sum $P\oplus Q =  \{x+y \mid x \in P, y\in Q\}$ has size $\Theta(p^2 q^2)$ in the worst case (see~\cite{AgarwalFH02}).  However, it is not clear that constructing the full arrangement is always necessary, and indeed Hernandez-Barrera~\cite{Barrera96algo} describes an $\tOh(pq)$ algorithm for the fixed-size case without rotations for \emph{orthogonal} $P$ and $Q$, avoiding a full construction. A similar case where a supposed arrangement lower bound was broken is for the Hausdorff distance under translation in $L_1$ and $L_\infty$~\cite{DBLP:conf/swat/ChewK92}.

This begs the question whether \emph{conditional lower bounds} from established hardness assumptions can be given.
Conditional lower bounds base their hardness on well-known and well-studied problems, which resisted algorithmic improvements despite significant efforts and are therefore conjectured not to be solvable in a certain running time, see~\cite{VassilevskaW18} for a recent survey. While these bounds can potentially be broken, this would require surprising breakthroughs (and even then, the obtained connections may reveal further information, such as which algorithmic techniques are required for further progress).
The only conditional lower bounds for the problem at hand that we are aware of are due to Barequet and Har-Peled~\cite{DBLP:journals/ijcga/BarequetH01} and they are based on the 3SUM Hypothesis~\cite{GajentaanO95}.\footnote{See Section~\ref{sec:prelim} for a definition of the $k$SUM problem and corresponding hardness assumptions.} Specifically, Barequet and Har-Peled show that any $\Oh(n^{2-\eps})$-time algorithm for translating an orthogonal polygon $P$ with $\Oh(n)$ vertices into an orthogonal polygon $Q$ with $\Oh(n)$ vertices would refute the 3SUM Hypothesis. This lower bound is then shown to hold also for translating and rotating a \emph{convex} $P$ into a \emph{convex} $Q$. While these results give first important insights into the complexity of the polygon placement problem, several questions remain unanswered: How do the different degrees of freedom affect the problems' hardness?  Can we establish (1) any superlinear bound for $p=\Oh(1)$ and (2) any hardness beyond $n^{2-o(1)}$? In this paper, we answer the last two questions affirmatively, and show that each degree of freedom (scaling, $x$-translation, $y$-translation and rotation) can be used to obtain a stronger hardness result, based on the $k$SUM conjecture. We complement these conditional lower bounds by an algorithmic result via offline dynamic algorithms for maintaining the union of rectangles.


\paragraph{Further Related Work.}
Obtaining conditional lower bounds for optimizing a geometric functional under translations of the input has recently received increasing interest. The complexity of the (discrete) Fr\'echet distance~\cite{AltG95,EiterM94}, a natural similarity measure of geometric curves, has been analyzed in~\cite{DBLP:journals/talg/BringmannKN21}: No algorithm can compute the discrete Fr\'echet distance under translations of polygonal input curves with at most $n$ vertices in time $\Oh(n^{4-\eps})$ under the Strong Exponential Time Hypothesis (SETH), or more precisely, the 4-OV Hypothesis. This lower bound matches a natural arrangement size lower bound, indicating that constructing such a large arrangement is inherent in the problem. Interestingly, the authors also give an $\Oh(n^{4.667})$-time algorithm by devising an offline dynamic algorithm for maintaining reachability in a grid graph. This theme of using offline dynamic algorithms for geometric optimization under translations will re-emerge in the present paper.  

There has been notable algorithmic work for optimizing the Fr\'echet distance subject to more general transformations: Wenk~\cite{wenk2002phd} proves that we can optimize the Fr\'echet distance over a class of of transformation with $d$ degrees of freedom in time $\Oh(n^{3d+2})$. To the best of our knowledge, no conditional lower bounds are known that prove a corresponding conditional hardness parameterized by the degrees of freedom. 
For minimizing the Hausdorff distance under translations, \cite{DBLP:journals/ijcga/BarequetH01} give 3SUM-based lower bounds for segments,  and \cite{BringmannN21} give 3SUM- and OV-based lower bounds for point sets.
Finally, motion planning problems related to polygon placement, such as moving a line segments through a given set of obstacles, are proven 3SUM-hard by Gajentaan and Overmars~\cite{GajentaanO95}.
Lower Bounds under $k$SUM for $k>3$ are much less prevalent than 3SUM lower bounds: Such lower bounds have been given, e.g., for problems related to convex hulls in $\mathbb{R}^d$~\cite{Erickson99}, exact subgraph finding~\cite{AbboudL13}, compressed string matching~\cite{AbboudBBK17} and compressed linear algebra~\cite{AbboudBBK20}.




\subsection{Our Results and Technical Overview}

To state our results, we first give an intuitive definition of $k$SUM and Orthogonal Vectors (OV), see Section~\ref{sec:prelim} for details. In the $k$SUM problem, we are given $k$ sets of $n$ integers, and we ask whether there are $k$ integers --- one from each set --- that sum to zero. This problem is hypothesized not to be solvable in time $\Oh(n^{\lceil k/2 \rceil - \epsilon})$ for any $\epsilon > 0$.
In the Orthogonal Vectors (OV) problem, we are given two sets of $n$ vectors in $\{0,1\}^d$ and we ask whether there is a vector from the first set that is orthogonal to a vector from the second set.
 This problem is hypothesized to not be solvable in time $\Oh(n^{2 - \epsilon}\mathrm{poly}(d))$ for any $\epsilon > 0$.

Throughout this section, let $p$ and $q$ denote the number of vertices of the input polygons $P$ and $Q$, respectively.
As a simple first result, we obtain that fixed-size polygon placement under translations has a quadratic-time hardness not only under the 3SUM Hypothesis (as shown in \cite{DBLP:journals/ijcga/BarequetH01}), but also under the OV Hypothesis which is well known to be implied by the Strong Exponential Time Hypothesis~\cite{Wil05}. This lower bound also applies to unbalanced cases of the problem, where $p$ and $q$ may differ significantly. 

\begin{thm} \label{thm:OVlowerbound}
	Assuming the OV Hypothesis, there is no $\Oh( (pq)^{1-\eps} )$-time algorithm for any $\eps > 0$ for finding an $x$-translation of an orthogonal polygon $P$ that fits into an orthogonal polygon $Q$. This holds even restricted to $p=\Theta(q^\alpha)$ for an arbitrary $\alpha > 0$. 
\end{thm}

\begin{wrapfigure}{r}{0.18\textwidth}
	\vspace{-.7cm}
  \begin{center}
    \includegraphics[width=0.18\textwidth]{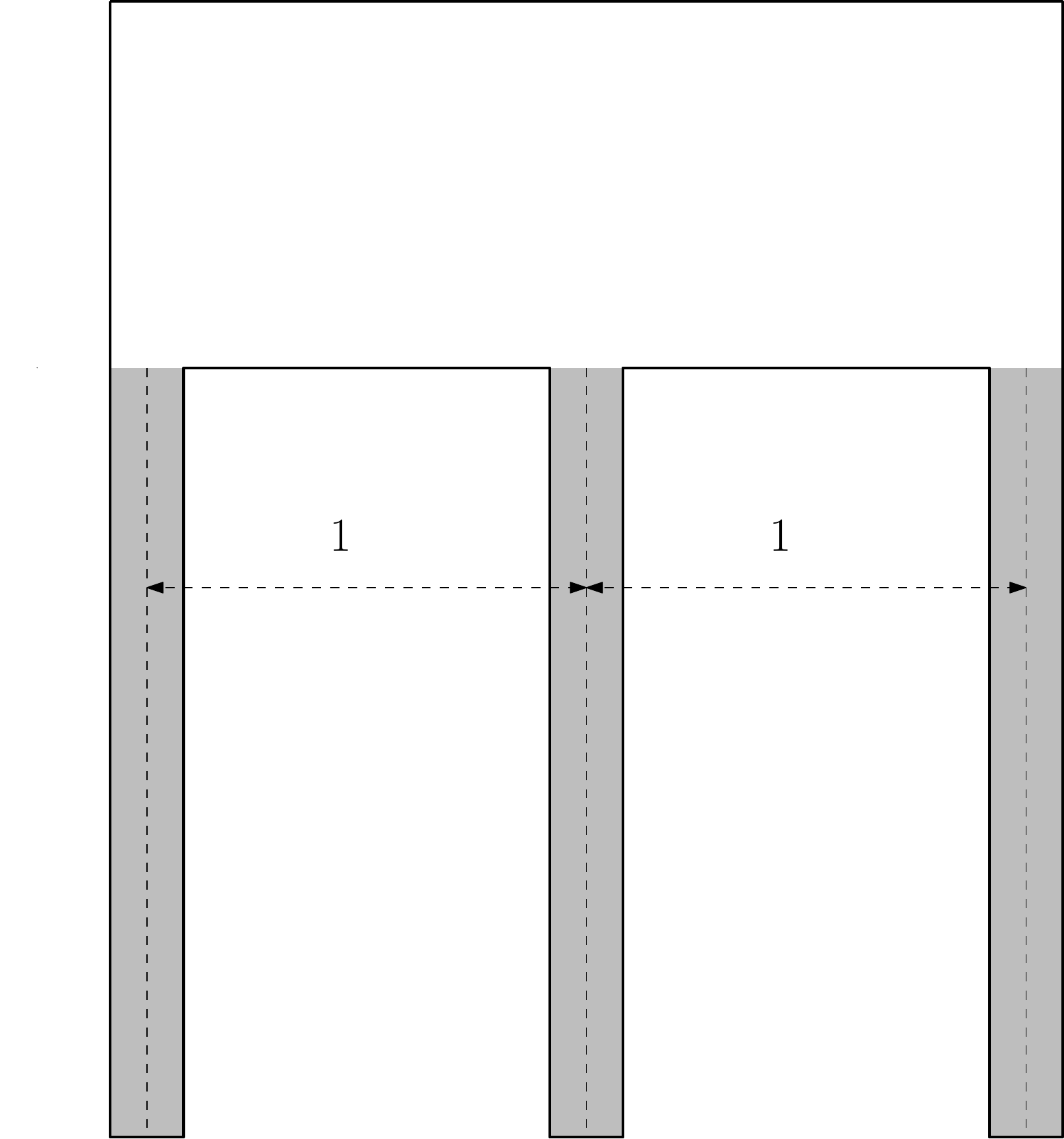}
  \end{center}
	\vspace{-.5cm}
\end{wrapfigure}
This lower bound tightly matches a corresponding $\Oh(pq \log(pq))$-time algorithm for orthogonal polygons $P$ and $Q$ due to Hernandez-Barrera~\cite{Barrera96algo}. Besides yielding a lower bound of $(pq)^{1-o(1)}$ for every relationship between $p$ and $q$, this also immediately proves the same hardness when we allow \emph{any polynomial-space} preprocessing of $P$ (or $Q$) by a recent result~\cite{AbboudVW21}.
In the remainder of this section, we turn to the cases where scaling is allowed. We first give a surprisingly simple 3SUM lower bound for translations in a single dimension and scaling.

\begin{thm}\label{thm:3SUM}
	Assuming the 3SUM Hypothesis, there is no $\Oh(q^{2-\eps})$-time algorithm for any $\eps > 0$ for finding a largest copy of an orthogonal $P$ that can be $x$-translated to fit into an orthogonal $Q$, even when $p = \Oh(1)$. 
\end{thm}

Specifically, we give quadratic-time hardness already when $P$ consists of a small number of $p=12$ vertices.  Technically, this result is achieved by a surprisingly simple reduction from the \average problem, which is defined as follows. Given a set $A$ of $n$ integers in $\{-U, \dots, U\}$ with $U=n^3$, determine whether there are $a_1,a_2,a_3\in A$ such that $a_2-a_1 = a_3-a_2$. Only recently, Dudek,  Gawrychowski, and Starikovskaya~\cite{DudekGS20} could prove 3SUM-hardness of this problem, resolving an open problem posed by Erickson.\footnote{As we shall see below, it is not strictly necessary to reduce from \average, since we can reduce from 3SUM directly if we exploit more of the geometric structure of the problem. However, reducing from this problem produces the simplest and most intuitive construction.} We reduce from this problem by using a polygon $P$ that essentially consists of three very long prongs, with an equal spacing of approximately 1 between them (see the figure on the top right).

\begin{wrapfigure}{l}{.4\textwidth}
	\vspace{-.6cm}
  \begin{center}
    \includegraphics[width=.4\textwidth]{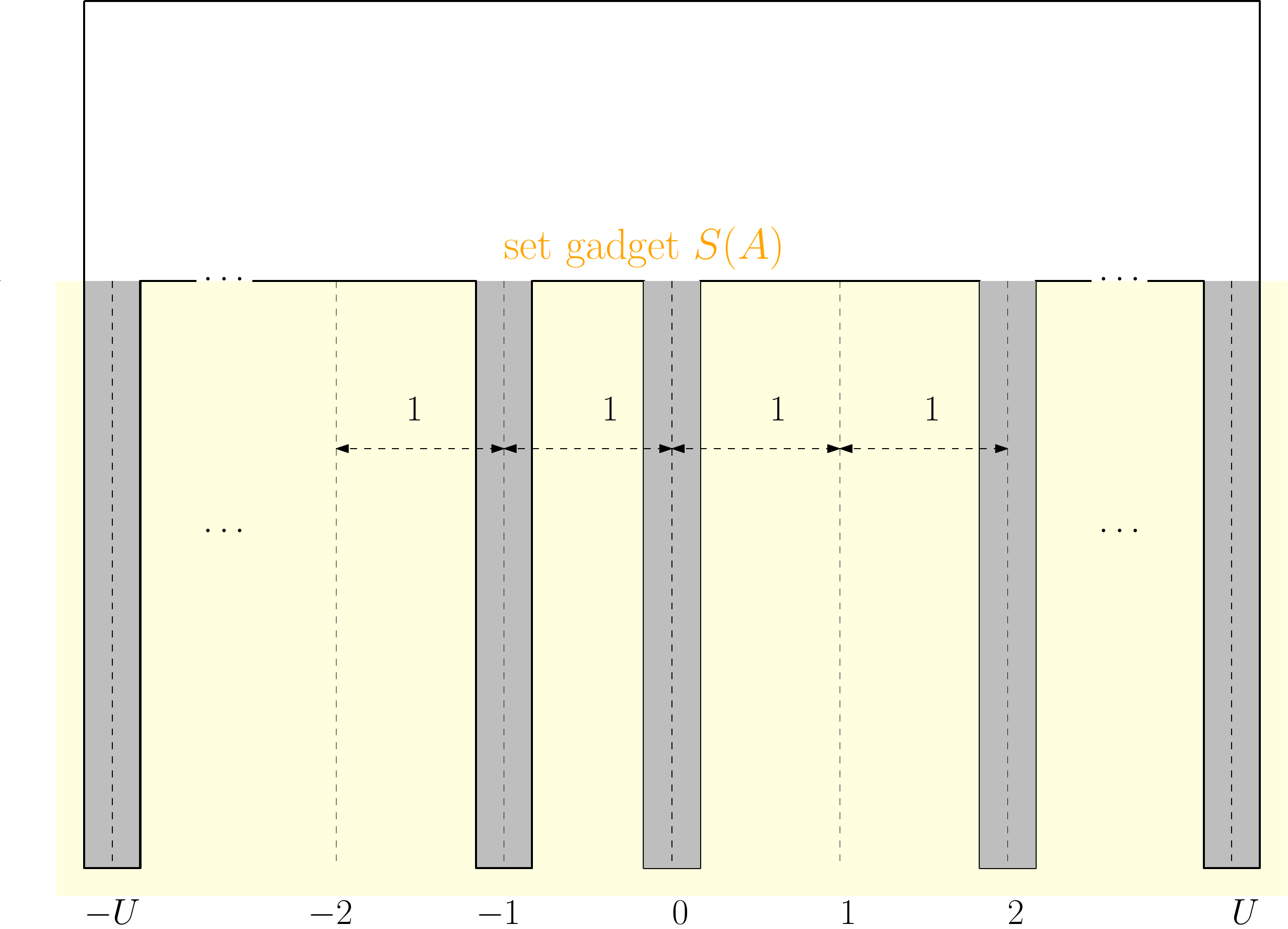}
  \end{center}
	\vspace{-.6cm}
\end{wrapfigure}
The corresponding polygon $Q$ is illustrated in the figure on the left: It encodes the set $A$ of input integers in $\{-U, \dots, U\}$ by long attached rectangles for each integer in $A$. The main idea is that we can place a scaled version of $P$ into $Q$ if and only if we can translate the left prong to be contained in a rectangle representing some $a_1\in A$, and choose a scaling of $\lambda \approx a_2 - a_1$ for some $a_2 \in A_2$ (making sure that the middle prong fits into the rectangle representing $a_2$) such that there exists $a_3\in A$ with $a_3 \approx a_2 + \lambda \approx 2a_2 - a_1$ (making sure that the right prong can be placed into the rectangle representing $a_3$). In other words, the instance is equivalent to the existence of $a_1,a_2,a_3 \in A$ with $a_2 - a_1 = a_3 - a_2$. We give the details (including how to choose widths and lengths of the prongs, etc.) in Section~\ref{sec:3SUM}. There, we also show that even \emph{approximating} the largest scaling factor by any polynomial factor suffers from the same hardness.

In the above lower bound, we only exploited scaling and $x$-translations to obtain our hardness, but ignored the possibility of $y$-translations. Indeed, taking into account this additional degree of freedom, we can obtain the same hardness result based on the 4SUM Hypothesis.

\begin{thm}\label{thm:4SUM}
	Assuming the 4SUM Hypothesis, there is no $\Oh(q^{2-\eps})$-time algorithm for any $\eps > 0$ for finding a largest copy of an orthogonal $P$ that can be $x$-translated to fit into an orthogonal $Q$, even when $p = \Oh(1)$. 
\end{thm}

It is well known that finding an $\Oh(n^{2-\eps})$-time algorithm for 4SUM would exponentially improve over the state-of-the-art meet-in-the-middle running time of $\Oh(2^{n/2})$ for Subset Sum.\footnote{This result follows by the split and list technique: Split the $n$ items into 4 parts of $n/4$ items each. For each part, list the weight of each subset of items, resulting in $4$ sets of at most $N=2^{n/4}$ numbers. If we could determine whether there exists a tuple $a_1,a_2,a_3,a_4$ with $a_1+a_2+a_3+a_4 = t$ in time $\Oh(N^{2-\eps})$, we would obtain an $\Oh(2^{N(1/2-\eps/4)})$-time algorithm for Subset Sum.}
Thus, the above lower bound also holds under the assumption that Subset Sum admits no exponential improvement over the meet-in-the-middle running time $\Oh(2^{n/2})$. Note that from the 3SUM lower bound, we could only deduce a $q^{1.5-o(1)}$ lower bound under the Subset Sum assumption.

\begin{wrapfigure}{r}{0.3\textwidth}
	\vspace{-.5cm}
  \begin{center}
    \includegraphics[width=0.3\textwidth]{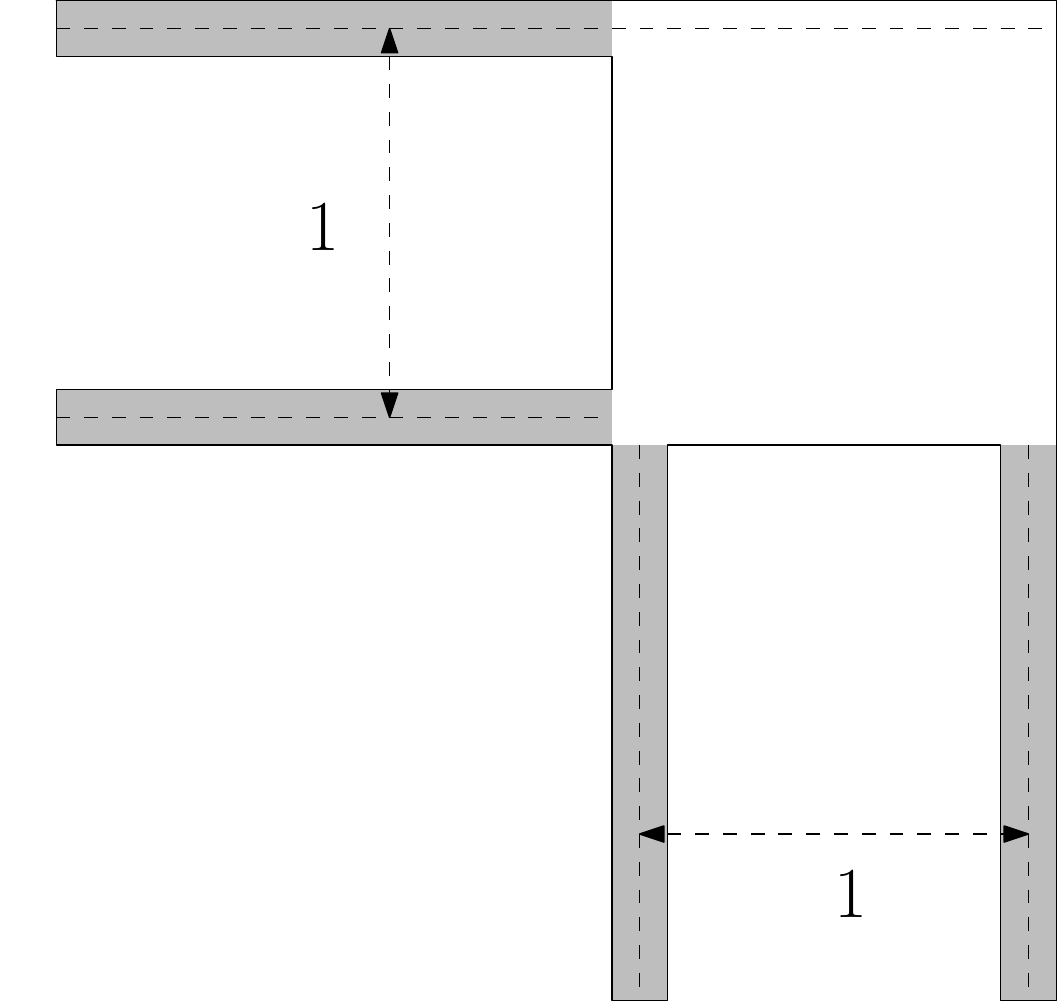}
  \end{center}
	\vspace{-.5cm}
\end{wrapfigure}
Technically, the result can be achieved similarly to the 3SUM lower bound. The first natural idea would be to design $P$ to have 2 horizontal and 2 vertical prongs, each pair within distance approximately $1$, see the figure on the right. It is intuitive that for $Q$, one can use the set gadget from the 3SUM lower bound on the bottom, representing a set $A$, and a vertical copy of the set gadget on the left, representing a set $B$, to encode the problem whether there exist distinct $a_1,a_2\in A$ and distinct $b_1,b_2\in B$ such that $a_1-a_2 = b_1-b_2$. Curiously, even with the techniques of Dudek et al.~\cite{DudekGS20}, it is unclear whether one can show 4SUM-hardness of the problem (the problem here is that not every summand is taken from a different set, i.e., the input is not 4-partite). However, exploiting the full generality of our geometric setting, we can circumvent this issue rather easily.     

\begin{wrapfigure}{l}{.3\textwidth}
	\vspace{-.5cm}
  \begin{center}
    \includegraphics[width=.3\textwidth]{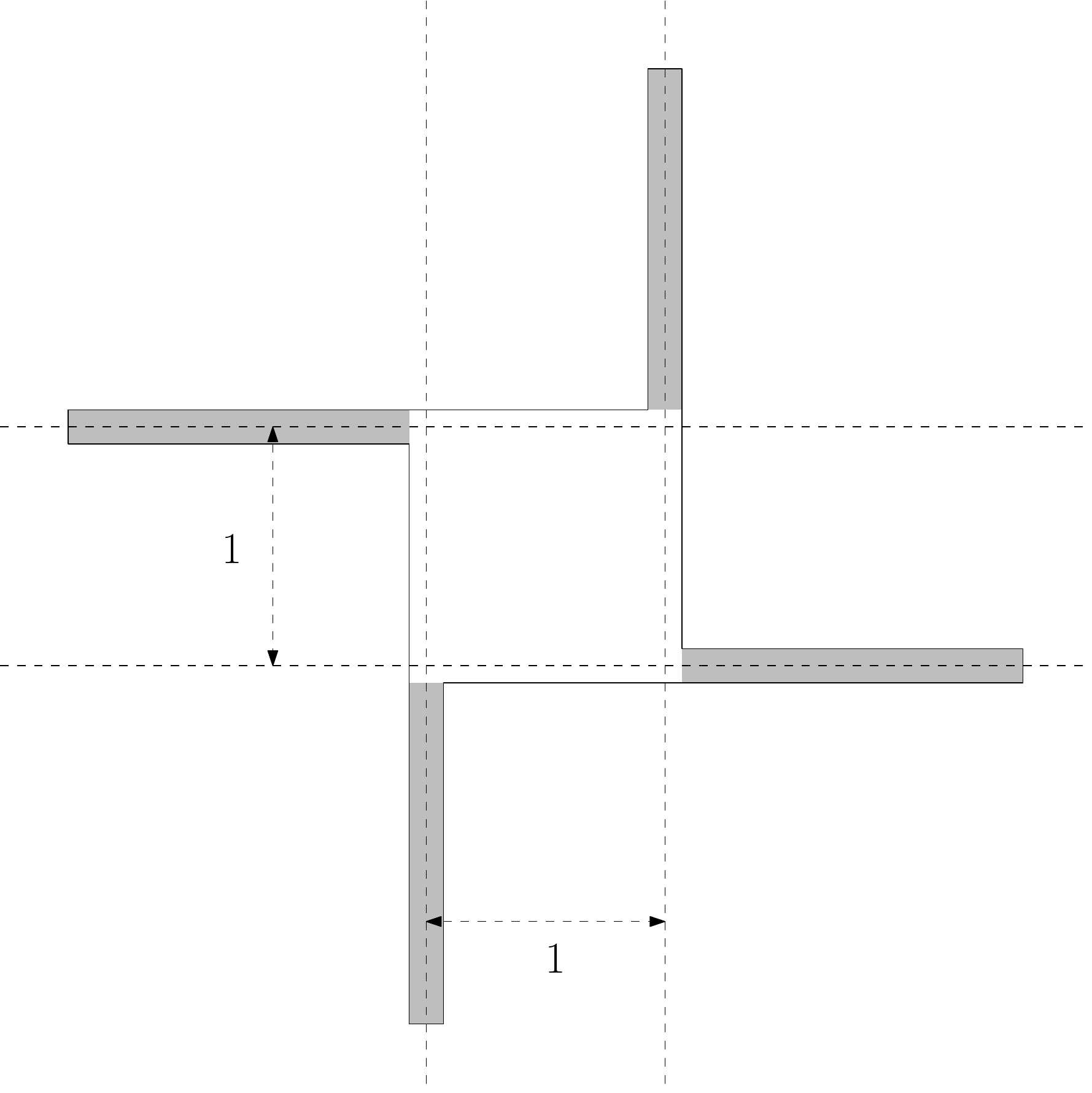}
  \end{center}
	\vspace{-.5cm}
\end{wrapfigure}
In the figure on the left, we illustrate the polygon $P$ that we use to obtain our result. It has two vertical prongs, one directed downwards, the other directed upwards, with a horizontal distance of approximately $1$. Likewise, we have a left horizontal and a right horizontal prong within vertical distance of approximately $1$. We choose $Q$ to be essentially square-shaped, with set gadgets on the bottom, left, top and right representing  
four input sets $A_1, B_1, A_2, B_2$, respectively. In this way, any placement of $P$ in $Q$ chooses some $a_1 \in A_1, b_1\in B_1$ as the position for the bottom and left prong, respectively. Note that the scaling factor $\lambda$ must be approximately $a_2 - a_1$ for some $a_2\in A_2$ so that the top prong can be contained in a rectangle representing $a_2$. Since $\lambda$ scales the distances of the horizontal and the vertical prongs uniformly, this also determines the fourth remaining summand: there must be some $b_2\in B_2$ such that $b_2 - b_1 \approx \lambda \approx a_2 - a_1$. Thus, there exists a (large) scaling of $P$ that can be placed into $Q$ if and only if there are $a_1\in A_1, a_2\in A_2, b_1\in B_1, b_2\in B_2$ such that $a_2 - a_1 = b_2 - b_1$.
To simplify the presentation, we defer the proof of the 4SUM lower bound to Section~\ref{sec:4SUM}, \emph{after} the presentation of the 5SUM lower bound in Section~\ref{sec:5SUM}.

\paragraph{Algorithmic Results and Tightness of Lower Bounds for Scaling and Translation.}
We show that our conditional lower bounds are not far from optimal for orthogonal polygons when $p$ or $q$ are constant.\footnote{Note that we formally prove hardness only for constant $p$. However, it is rather straightforward to adapt the lower bounds to constant $q$, which we will do in the full version of the paper.} To this end, let $\lambda P$ denote the polygon $P$ scaled by $\lambda$. It is not difficult to obtain an $\tOh((pq)^3)$ baseline algorithm: We first observe that there is a set of at most $\Oh((pq)^2)$ \emph{critical} values that must contain the largest scaling factor $\lambda^*$.\footnote{For readers' convenience, we give an intuitive argument: Take an optimal scaling $\lambda^*$ of $P$ and fix a corresponding placement of $\lambda^* P$ into $Q$. Move the polygon to the left and to the bottom until there is contact between horizontal segments $e_P, e_Q$ of $\lambda^* P$ and $Q$ as well as vertical segments $f_P,f_Q$ of $\lambda^* P$ and $Q$. The claim is that there exists an additional pair of horizontal or vertical segments $g_P,g_Q$ of $\lambda^* P$ and $Q$ that have a contact. Otherwise, we could use a larger scaling factor $\lambda'> \lambda^*$ and still be able to place $\lambda' P$ into $Q$. Now, if $g_P,g_Q$ are horizontal segments, then the contacts of $e_P,e_Q$, and $g_P,g_Q$ uniquely determine $\lambda^*$, otherwise the contacts of $f_P,f_Q$ and $g_P,g_Q$ uniquely determine $\lambda^*$. Thus, there are at most $2(pq)^2$ possible scalings, given by horizontal or vertical contacts.} For each such critical scaling factor $\lambda$, we can determine whether $\lambda P$ can be translated into $Q$ in time $\Oh(pq\log(pq))$ by~\cite{Barrera96algo}. At this level of abstraction, this algorithm cannot be improved: Even with a fixed $Q$ with $q$ vertices, testing for a list of $L$ polygons $P_1, \dots, P_L$ with $p$ vertices whether $P_i$ can be translated to be contained in $Q$ can be proven to require time $(Lpq)^{1-o(1)}$ under the OV Hypothesis (via an adaptation of Theorem~\ref{thm:OVlowerbound}).

However, we can improve this algorithm significantly, yielding  an $\tOh(q^{2.5})$-time algorithm for constant $p$.

\begin{thm}\label{thm:algo}
For polygons $P,Q$, let $\lambda^*$ denote the largest scaling factor $\lambda$ such that $\lambda P$ can be placed into $Q$ under translations.
Given any orthogonal simple polygons $P$ and $Q$ with $p$ and $q$ vertices, respectively, we can compute $\lambda^*$ in time $\Oh((pq)^{2.5}2^{\Oh(\log^* pq)})$.
\end{thm}

Interestingly, the above algorithm is obtained using an offline dynamic algorithm for maintaining the union of $n$ rectangles in amortized $\tOh(\sqrt{n})$ time, which led to algorithms for Klee's measure problem (see~\cite{OvermarsY91,Chan10}). The basic idea is to cover $P$ and $\mathbb{R}^2 \setminus Q$ by rectangles $P_1, \dots, P_{p'}$ and $Q_1, \dots, Q_{q'}$, respectively, where $p' = \Oh(p)$ and $q' = \Oh(q)$. Let $R_{i,j}$ denote the set of translations of $P$ such that $P_i$ intersects the interior of $Q_j$, yielding a set of \emph{forbidden} translations (under such a translation, some part of $P$ lies outside of $Q$). Note that the $R_{i,j}$ are rectangles as well.
Consequently, $P$ can be placed into $Q$ under translation $\tau$ if and only if $\tau$ is not contained in any of the forbidden regions, i.e., it is not contained in $\bigcup_{i,j} R_{i,j}$. 

For a scaled copy $\lambda P$, we view the set of $R_{i,j}$'s as a dynamically changing set of (open) rectangles, parameterized by the scaling factor $\lambda$. While these rectangles may completely change with every change in~$\lambda$, we may find a suitable representation in rank space that incurs only few, i.e., $\Oh((pq)^2)$, combinatorial changes. More specifically, we show how to reduce the problem of determining the largest $\lambda$ such that $\bigcup_{i,j} R_{i,j}(\lambda)$ does not cover all possible translations to the problem of determining the first point at which a dynamically changing set of rectangles contains a hole, where updates are specified in advance. Using a data structure with amortized $\tOh(\sqrt{n})$-time updates, we obtain a $\tOh((pq)^{2.5})$-time algorithm over all $\Oh((pq)^2)$ updates. We give the proof in Section~\ref{sec:algo}. 

\paragraph{The Remaining Degree of Freedom: Rotation.}
Finally, we take into account rotation as an additional degree of freedom. Specifically, for the setting of scaling, translation, and rotation, we obtain cubic-time hardness based on the 5SUM Hypothesis.

\begin{thm}\label{thm:5SUM}
	Assuming the 5SUM Hypothesis, there is no $\Oh((p+q)^{3-\eps})$-time algorithm for any $\eps > 0$ for finding a largest copy of a simple polygon $P$ that can be rotated and translated to fit into a simple polygon $Q$.
\end{thm}

This result concludes our line of $(\text{degree of freedom} + 1)$-SUM hardness results for the problem. For each degree of freedom, we show how it can be used to determine one summand in a corresponding $k$SUM instance, where the remaining summand is uniquely determined by the previous choices. We are not aware of any other geometric optimization problem for which such a natural $(\text{degree of freedom} + 1)$-SUM hardness is known.

\begin{wrapfigure}{r}{0.3\textwidth}
	\vspace{-.5cm}
  \begin{center}
    \includegraphics[width=0.3\textwidth]{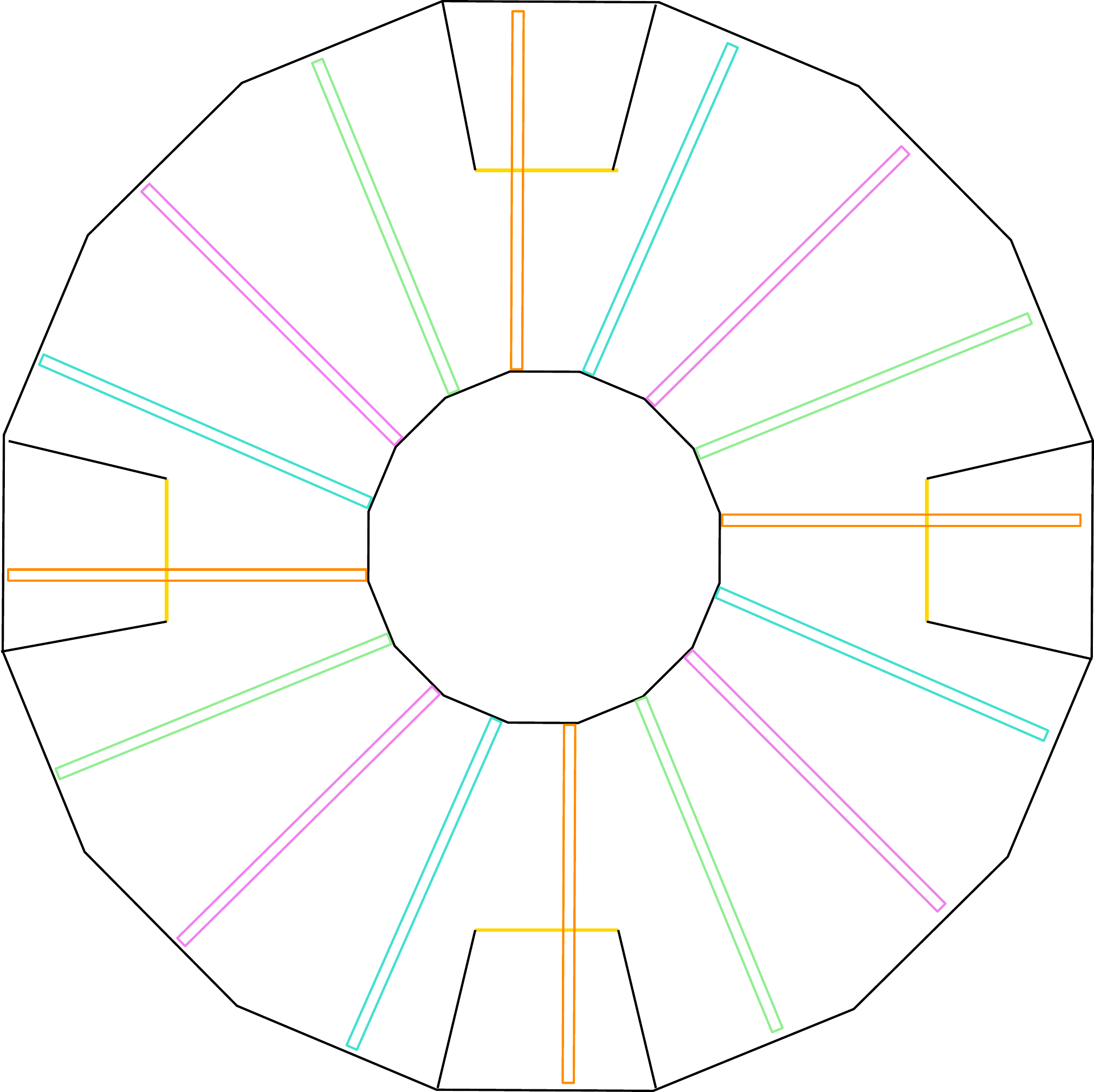}
  \end{center}
	\vspace{-.5cm}
\end{wrapfigure}
The 5SUM lower bound of Theorem~\ref{thm:5SUM} is perhaps the most interesting technical contribution of the paper. A natural approach to adapt the 4SUM construction would be to arrange many rotated copies of the four prongs of $P$ into a wheel-like structure (see the figure on the right). In such a construction, the rotation can choose a group $i$ of four prongs to be tested against set gadgets situated at the bottom, left, top, and right. Which information can we encode into the choice of the group? A central idea of our approach is that, instead of letting the prongs have a uniform distance of approximately 1, we use different distances for the corresponding $x$- and the $y$-directions, say a distance of $1$ for the vertical prongs and a distance of $\mu_i$ for the horizontal prongs. With such a construction, one might hope to reduce from the following problem:
\[
\exists a_1 \in A_1, a_2\in A_2, b_1\in B_1, b_2\in B_2, \mu \in M: b_2 - b_1 = \mu (a_2 - a_1).
\]
While this problem has a similar structure to 5SUM, it is unclear whether it is indeed 5SUM-hard: the combination of additive and multiplicative structures might make this problem simpler. It turns out, however, that we can reduce from 5SUM into our geometric problem directly, by exploiting that our gadgets can be made robust against small perturbations (we need to ensure a certain robustness anyway to deal with slight possible variations of the desired translations, scalings, and rotations, which poses a minor technical complication in our proofs). More specifically, the idea is to use that, when choosing $\mu = 1+\eps$ for small values of $\eps$, we can approximate $\mu(a_2-a_1) = (1+\eps)(a_2-a_1) \approx a_2 - a_1 + \eps$. To exploit this trick, consider the following (equivalent) formulation of 5SUM: 
\[ \exists a_1 \in A_1, a_2\in A_2, b_1\in B_1, b_2\in B_2, c \in C: b_2 - b_1 = a_2 - a_1 + c,\]
where $A_1,A_2,B_1,B_2,C$ are sets of $n$ integers in $\{-U, \dots, U\}$ with $U=n^5$. The idea is to set $\mu_i = 1+\frac{c_i}{M}$ where $c_i$ is the $i$-th element in $C$ and $M$ is a large value to be determined later. We need to adapt the placement of the set gadgets slightly: the bottom and top set gadgets, representing $A_1, A_2$, should be spaced at a horizontal distance of essentially $M$. Likewise, we place the left and right set gadgets, representing $B_1$ and $B_2$, at a vertical distance of essentially~$M$. Intuitively, we can place $\lambda P$ (for large enough $\lambda$) into $Q$ if and only if the following conditions hold:
\begin{itemize}
\renewcommand\labelitemi{--}
\item The rotation chooses a group and thus a corresponding factor $\mu_i = 1+\frac{c}{M}$ for some $c\in C$.
\item The $x$-translation chooses some $a_1\in A_1$ whose rectangle contains the bottom prong of the chosen group.
\item The $y$-translation chooses some $b_1\in B_1$ whose rectangle contains the left prong of the chosen group.
\item The scaling $\lambda$ is approximately equal to $M + a_2-a_1$ for some $a_2\in A_2$ whose rectangle contains the top prong of the chosen group.
\item Finally, the right prong of the chosen group must be contained in the rectangle of some $b_2\in B_2$ determined by $ M + b_2 - b_1 \approx \lambda \mu \approx (M+ a_2 - a_1) (1+\frac{c}{M}) = M + a_2 - a_1 + c + \frac{(a_2-a_1)c}{M}$. By choosing $M$ large enough (e.g., $M = U^3$) we can ensure that the approximation error $\frac{(a_2-a_1)c}{M} = o(1)$ is negligible. We obtain equivalence with $b_2 - b_1 = a_2 - a_1 + c$.
\end{itemize}
In Section~\ref{sec:5SUM}, we give the proof of the theorem, showing how to also overcome further technical complications, such as possible deviations from the intended translations, scalings, and rotations (e.g., we need to enforce that a rotation is chosen that enforces prongs intended for $A_1,A_2$ are not rotated to intersect with $B_1,B_2$).

\paragraph{Approximation Hardness.} It turns out that the above hardness results also give strong hardness of approximation. We detail this for the case of scaling and x-translations in Section~\ref{sec:3SUM}. The same arguments can be used to obtain similar statements for our other settings. 

\subsection{Outlook and Open Problems}

We provide insights into the fine-grained complexity of the polygon placement problem, by showing $(\text{degree of freedom} + 1)$-SUM hardness. For the cases without rotations and constant-size~$P$, we show that the optimal running time is between $q^{2-o(1)}$ (assuming the 4SUM Hypothesis) and $\tOh(q^{2.5})$, by an improvement via offline dynamic algorithms. Can we close this gap? We believe that an algorithmic improvement is more likely than a higher lower bound.

Similar challenges remain for more general cases: Can we generalize some of our algorithmic ideas for orthogonal polygons to general simple polygons? On the hardness side, it is an interesting challenge to prove even higher, i.e., super-cubic, conditional lower bounds when $p = \Theta(q)$. However, proving higher bounds from the $k$SUM Hypothesis appears rather unlikely, as in our constructions, each degree of freedom naturally chooses a corresponding summand in a given $k$SUM instance.
Finally, can we use ideas from this paper to give hardness also for convex $P$ and/or $Q$? Barequet and Har-Peled~\cite{DBLP:journals/ijcga/BarequetH01} could adapt their lower bounds for the fixed-size case under translation to convex $P$ and $Q$ in a quite natural way. It is unclear how this could be done for our lower bounds, so that proving stronger lower bounds for convex restrictions remains an interesting open problem.

%
%
%
%

\section{Preliminaries} \label{sec:prelim}

Before we can show our results, we introduce some notation, cleanly define the problems at hand, and revisit some geometric facts that will come in handy in our proofs.

\subsection{Notation and Conventions}

We let $[n]$ denote $\{1, \dots, n\}$. 
For any $x,y,\eps \in \mathbb{R}$, we use $x \in y + [-\eps, \eps]$ to denote $x \in [y - \eps, y + \eps]$. We typically think of a placement of $P$ into $Q$ via a tuple $(\lambda, \tau, \alpha)$: we scale $P$ by the scaling factor $\lambda > 0$ (with respect to the center of $P$ as its reference point), rotate it in counter-clockwise direction by the rotation angle $\alpha \in (-\pi, \pi]$ (around the center of $P$) and translate the resulting polygon by the translation vector $\tau \in \mathbb{R}^2$. To this end, let $\lambda P$ denote the polygon $P$ scaled by $\lambda$.

\subsection{Hardness Assumptions}

In this work we use two standard hypotheses from fine-grained complexity theory: the $k$SUM Hypothesis and the OV Hypothesis. Most of our lower bounds are derived from the former.

\subsubsection*{$k$SUM.}
For $k\ge 2$, let $k$SUM denote the following problem: Given sets $A_1, \dots, A_k$ of $n$ integers in $\{-U, \dots, U\}$ with $U=n^k$, determine whether there exist $a_1\in A_1, \dots, a_k \in A_k$ with $\sum_{i=1}^k a_i = 0$. This problem is well known to be solvable in time $\Oh(n^{\lceil k/2 \rceil})$. The hypothesis that this running time cannot be improved by a polynomial factor is known as the $k$SUM Hypothesis:  

\begin{hypo}[$k$SUM Hypothesis]\label{hypo:kSUM}
	Let $k\ge 3$. For no $\epsilon > 0$, there is an $\Oh(n^{\lceil k/2 \rceil-\epsilon})$-time algorithm for $k$SUM.
\end{hypo}
There is a vast literature of conditional lower bounds based on the 3SUM conjecture, starting with~\cite{GajentaanO95} (see~\cite{VassilevskaW18} for an overview). For lower bounds from $k$SUM for $k>3$, we refer to~\cite{Erickson99, AbboudL13, AbboudBBK17, AbboudBBK20}.

\subsubsection*{Orthogonal Vectors.}
We also give a simple lower bound based on the Orthogonal Vectors (OV) problem. In this problem, we are given two sets $A,B$ of vectors in $\{0,1\}^d$, and the task is to determine whether there is a pair $a\in A, b\in B$ that is orthogonal, i.e., for all $k\in[d]$ we have $a_i[k]\cdot b_j[k] = 0$. The obvious algorithm runs in time $O(|A| |B|d)$, and is conjectured to be best-possible up to subpolynomial improvements in $|A|,|B|$:
\begin{hypo}[OV Hypothesis]\label{hypo:OV}
	There is no $\eps > 0$ and $\alpha > 0$ such that we can solve OV with $|A| = \Theta(|B|^{\alpha})$ in time $\Oh( (|A|\cdot |B|)^{1-\eps}\mathrm{poly}(d))$.
\end{hypo}
This hypothesis is typically stated for fixed $\alpha=1$, which is well-known to be equivalent to the above formulation, see e.g. \cite[Lemma II.1]{BringmannK15}. It is well known that the OV Hypothesis is implied by the Strong Exponential Time Hypothesis, see~\cite{Wil05, VassilevskaW18}.

\subsection{Basic Geometric Facts}

We exploit the following approximations of trigonometric functions.
\begin{fact}[Small angle approximations]\label{fact:trig-approx}
	For $0 \le \alpha \le \frac{\pi}{4}$, we have
	\begin{align*}
		\frac{2}{\pi} \alpha & \le \sin(\alpha) \le \alpha, & 1 - \frac{\alpha^2}{2} & \le \cos(\alpha) \le 1, & \alpha & \le \tan(\alpha) \le \alpha + \alpha^3 
	\end{align*}
\end{fact} 

The basic shape of our polygons will be regular $4n$-gons, for which we collect the following facts.

\begin{lem}[Regular $4n$-gon facts]\label{fact:polygons}
	Let $a_P$ and $R_P$ denote the apothem and circumradius, respectively, of a regular $4n$-gon $P$ of side length $s$. Then for sufficiently large $n$,
	\[ \frac{sn}{2} \le a_P \le R_P \le sn. \]
\end{lem}
\begin{proof}
	Note that by Fact~\ref{fact:trig-approx}, $\tan(\frac{\pi}{4n}) \le \frac{\pi}{4n} + (\frac{\pi}{4n})^3 \le \frac{1}{n}$ for sufficiently large $n$. Thus, $a_P = (s/2) \cdot (1/\tan(\frac{\pi}{4n})) \ge sn/2$. Furthermore, we have $a_P\le R_P = (s/2) \cdot (1/\sin(\frac{\pi}{4n})) \le (s/2) (\pi/2) (4n/\pi) = sn$.
\end{proof}

Finally, the following lemma enables us to conveniently analyze intersections of the rotated polygon when we know its rotation up to a small additive error in $[-\eps,\eps]$. 

\begin{figure}
  \begin{center}
    \includegraphics[width=0.4\textwidth]{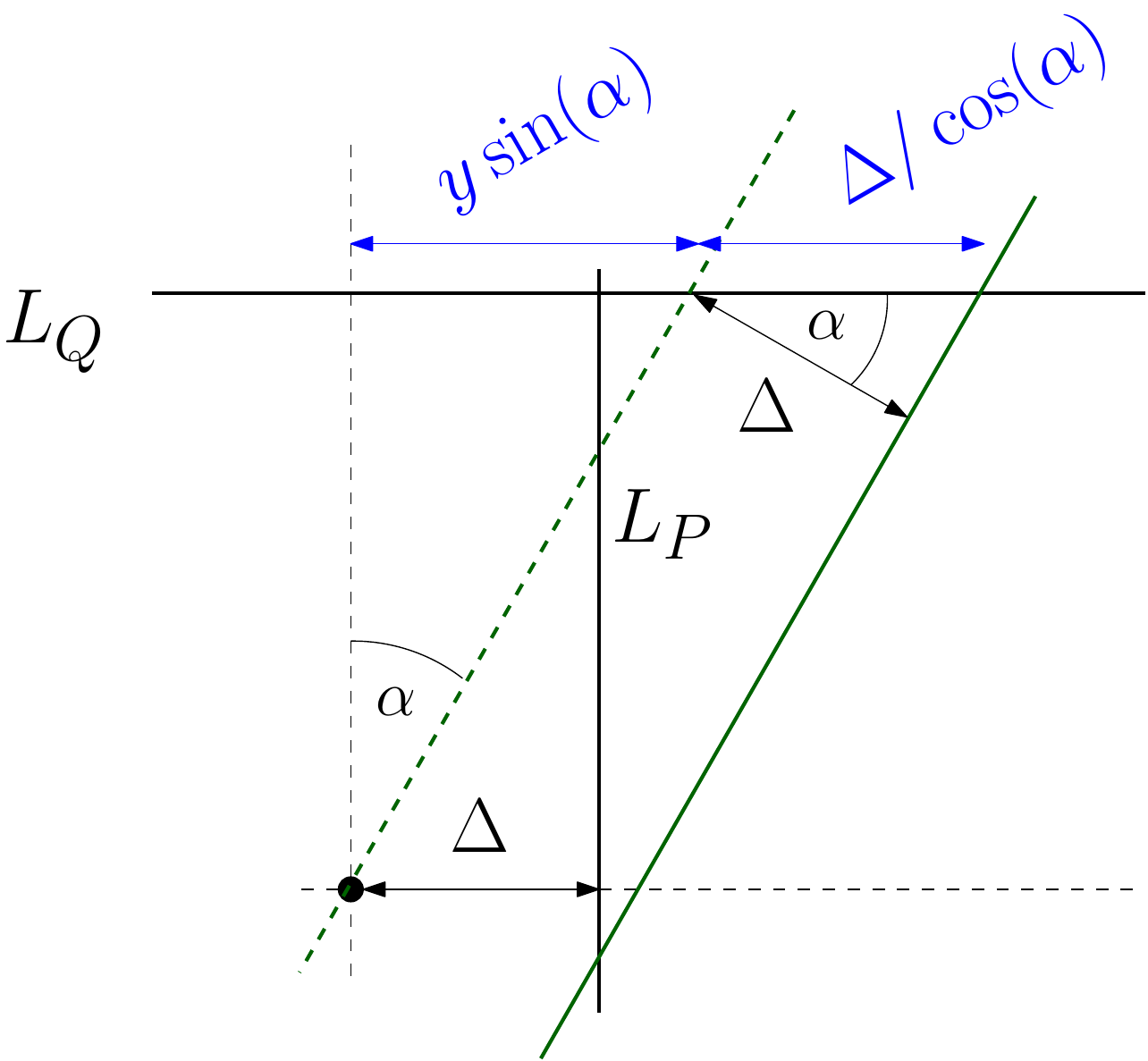}
  \end{center}
  \caption{Intersection of a horizontal line with a vertical line at distance $\Delta$, rotated by $\alpha$}
	\label{fig:line-rotation}
\end{figure}

\begin{lem}[Line Rotation Lemma]\label{lem:line-rotation}
	Consider a vertical line $L_P$ at $x=\Delta$ and a horizontal line $L_Q$ at $0 \le y\le Y$. Let $\Delta'$ denote the $x$-coordinate at which $L_P$, rotated around the origin by an angle $\alpha \in [-\eps, \eps]$ with $\eps \le \pi/4$, intersects $L_Q$. Then \[\Delta - Y\eps \le \Delta' \le \Delta(1+\eps^2) + Y \eps.\]
\end{lem}
\begin{proof}
	As illustrated in Figure~\ref{fig:line-rotation}, the intersection of the rotated $L_P$ with $L_Q$ is at the $\x$-coordinate $\Delta/\cos(\alpha) + y\sin(\alpha)$. By Fact~\ref{fact:trig-approx}, we obtain
	\[ \Delta/\cos(\alpha) + y \sin(\alpha) \ge \Delta - Y \eps,\]
	as well as, using $(1-\alpha^2/2)^{-1} \le 1+ \alpha^2$ for $\alpha \le 1$,
	\[\Delta/\cos(\alpha) + y \sin(\alpha) \le \Delta(1-\alpha^2/2)^{-1} + Y\eps \le \Delta(1+ \eps^2) + Y\eps. \]
\end{proof}

\section{OV Lower Bound for Translation}
\label{sec:OV}

In this section, we present the OV-based lower bound for the polygon placement problem with translations only, i.e., we prove Theorem~\ref{thm:OVlowerbound}.
Given $A, B \subseteq \{0,1\}^d$, we construct $P$ and $Q$ as illustrated in Figures~\ref{fig:OVp} and~\ref{fig:OVq}, respectively. For the polygon $P$, we define $|A|$ vector gadgets: the vector gadget for $a_i\in A$ has width $d$, and the $k$-th width-$1$ section of the gadget has height $1 + a_i[k]$. To obtain the full polygon $P$,  we concatenate the vector gadgets for $a_1, \dots, a_{|A|}$, separated by width-1, height-3 \emph{separator} rectangles (which we also prepend to the left of $a_1$ and append to the right of $a_{|A|}$, see Figure~\ref{fig:OVp}. In total, we obtain a polygon of width $W \coloneqq |A|(d+1) + 1$ and height 3, consisting of $\Oh(|A|d)$ vertices.

For the polygon $Q$, we define $|B|$ analogous vector gadgets: the vector gadget for $b_j\in B$ has width $d$, and the $k$-th width-$1$ section of the gadget has height $2 - b_j[k]$. Again, we obtain the full polygon $Q$ by concatenating the vector gadgets for $b_1, \dots, b_{|B|}$, this time separated by width-$\Delta$, height-3 separator rectangles with $\Delta \coloneqq (|A|-1)(d+1) + 1$, which we also prepend to the left of $b_1$'s vector gadget and append to the right of $b_{|B|}$'s vector gadget. This way, we obtain a polygon $Q$ with $\Oh(|B|d)$ vertices.

\begin{figure}

\centering
\includegraphics[width=0.8\textwidth]{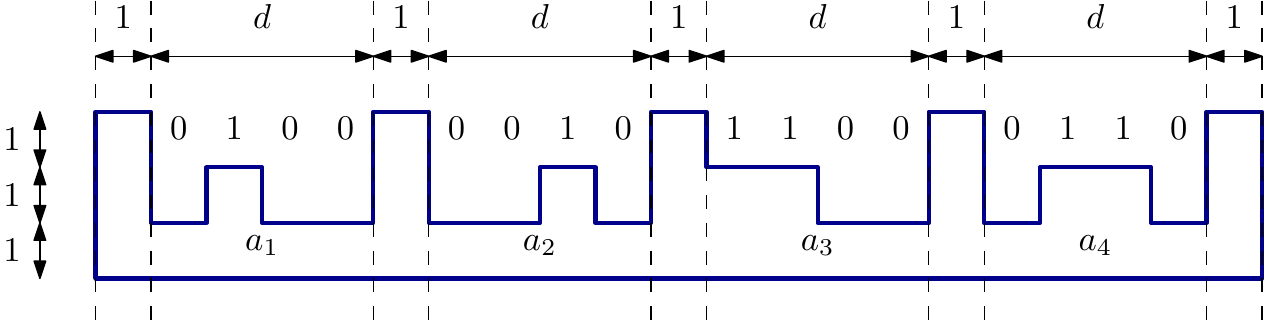}
\caption{The polygon $P$ (blue) capturing a vector set $A = \{0100, 0010, 1100, 0110\}$.}
\label{fig:OVp}

\vspace*{\floatsep}

\includegraphics[width=\textwidth]{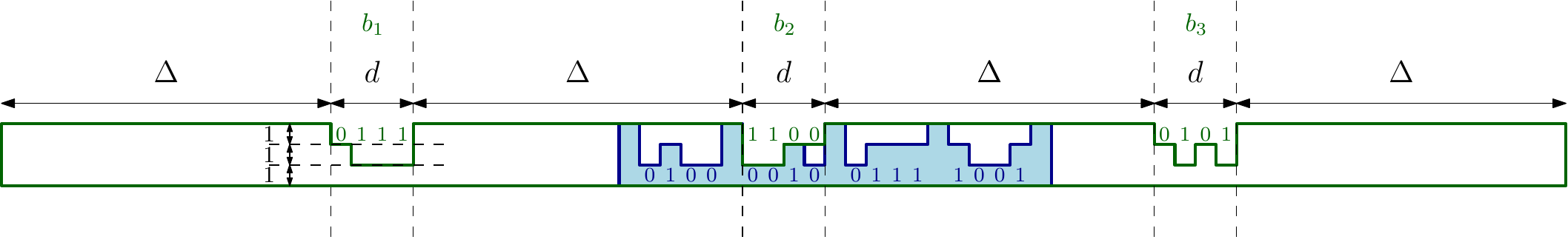}
	\caption{The polygon $Q$ (green) capturing the vector set $B = \{0111, 1100, 0101\}$. We also give a feasible placement of a polygon $P$ (blue) for a vector set $A=\{0100, 0010, 0111, 1001\}$. Recall that $\Delta = (|A|-1)(d+1) + 1$.}
\label{fig:OVq}

\end{figure}

\begin{claim}
	$P$ can be translated to be contained in $Q$ if and only if there are $a_i \in A, b_j\in B$ that are orthogonal.
\end{claim}
\begin{proof}
	Assume that there is an orthogonal pair $a_i\in A, b_j \in B$. Then we place $P$ into $Q$ by having their bottom boundaries align and using the unique $x$-translation that aligns the vector gadgets of $a_i$ and $b_j$ (see Figure~\ref{fig:OVq} for an example). The vector gadget of $a_i$ fits into the vector gadget for $b_j$, since the $k$-th width-1 section of $a_i$'s vector gadget has height $1+a_i[k] \le 2- b_j[k]$ (as orthogonality of $a_i,b_j$ implies $a_i[k]+b_j[k] \le 1$ for all $k$). Furthermore, the part of $P$ that is to the left of $a_i$'s vector gadget has width $(i-1)(d+1) + 1 \le \Delta$ and height 3, and thus fits into the separator to the left of $b_j$'s vector gadget. Likewise, the part of $P$ to the right of $a_i$'s vector gadget fits into the separator to the right of $b_j$'s vector gadget.

	Conversely, consider any translation of $P$ that fits into $Q$. We claim that it must perfectly align some vector gadget $a_i$ with some vector gadget $b_j$: Since the leftmost and rightmost separator rectangle of $P$ have a distance of $|A|(d+1) - 1 > \Delta$, these rectangles must be placed in two different separator rectangles of $Q$. Fix $b_j$ to be the vector represented between these two separator rectangles. Since this vector gadget has height at most $2$, we must find, in this width-$d$ section, a corresponding width-$d$ section of $P$ of height at most $2$, which can only be the vector gadget for some $a_i$ (by the structure of $P$). Since the vector gadget of $a_i$ must fit into the vector gadget for $b_j$, we obtain from containment in the $k$-th width-$1$ section that $1+ a_i[k] \le 2-b_j[k]$. Since $a_i[k], b_j[k] \in \{0,1\}$, we obtain that $a_i[k] \cdot b_j[k] = 0$ for all $k$, proving that $a_i$ and $b_j$ are orthogonal.    
\end{proof}

From this reduction, a $(pq)^{1-o(1)}$ lower bound based on the OV Hypothesis is immediate.

\begin{proof}[Proof of Theorem~\ref{thm:OVlowerbound}]
	Assume for contradiction that there are $\eps, \alpha >0$ such that we can solve the polygon placement problem under translation for $Q$ and $P$ with $q$ and $p=\Theta(q^\alpha)$ vertices in time $\Oh((pq)^{1-\eps})$. Then, given any OV instance $A,B$ with $|A|=\Theta(|B|^\alpha)$, using the above reduction we can produce an equivalent polygon placement problem instance $P$, $Q$ with $p = \Theta(|A|\mathrm{poly}(d))$, $q= \Theta(|B|\mathrm{poly}(d))$ and at the same time $q=\Theta(p^\alpha)$.\footnote{To be precise, we might need to add a small number of vertices to $P$ or $Q$ without changing the polygons' shapes to achieve the desired $p=\Theta(q^\alpha)$.} Since we can solve this instance in time $\Oh((pq)^{1-\eps}) = \Oh((|A||B|)^{1-\eps}\mathrm{poly}(d))$ by assumption, this would refute the OV Hypothesis (Hypothesis~\ref{hypo:OV}).
\end{proof}

\section{\boldmath 3SUM Lower Bound for Scaling and $x$-Translation}
\label{sec:3SUM}

\begin{wrapfigure}{r}{.38\textwidth}
	\vspace{-.5cm}
  \begin{center}
    \includegraphics[width=0.38\textwidth]{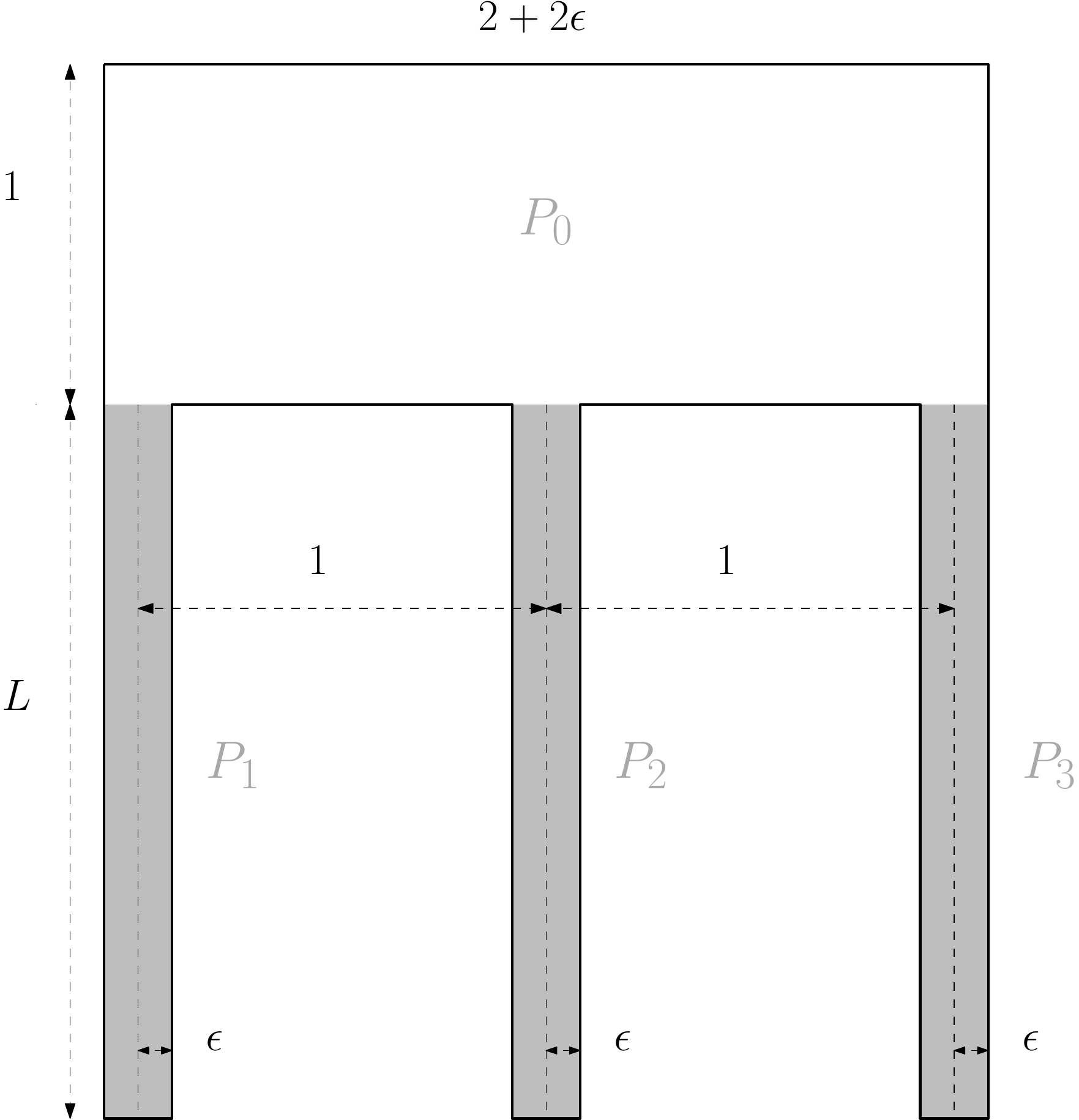}
  \end{center}
  \caption{Polygon $P$ of the reduction from 3SUM.}
	\label{fig:3SUMp}
	\vspace{-.5cm}
\end{wrapfigure}
In this section we present the reduction from 3SUM to the polygon placement problem with scaling and $x$-translations. We first describe the reduction and proceed to proving its correctness. Finally, we also show that polynomial approximations can be ruled out using the same reduction.

Let $A = \{a_1, \dots, a_n\} \subseteq \{-U, \dots, U\}$ with $U=n^3$ be a given \average instance.
Consider the polygons $P$ and $Q$ depicted in Figures~\ref{fig:3SUMp} and~\ref{fig:3SUMq} and let $L$, $L'$, $\eps$, and $\delta$ be parameters that we set later.
$P$ is a $(2+2\eps)\times 1$ rectangle $P_0$ with three attached prongs $P_1, P_2, P_3$, where the prongs have length $L$, width $2\eps$ and are evenly spread (such that the centers of neighboring prongs have a distance of $1$). The polygon $Q$ encodes the set $A$: It is a $(2U+2\delta)\times U$ rectangle $Q_0$ with $n$ attached prongs $Q_1, \dots, Q_n$ whose placement is determined by $A$. Specifically, consider $2U+1$ evenly spaced segments of width $2\delta$ along the bottom boundary representing the universe elements $\{-U, \dots, U\}$ (see Figure \ref{fig:3SUMq} for details). For each $a_i \in A$, we attach the prong $A_i$ of length $L'$ and width $2\delta$ to the segment representing $a_i$. We choose the values of the parameters as 
\begin{align*}
	L &= 2U, & \eps &= \frac{1}{10 U}, & \\
	L' &= UL, & \delta  &= U \eps.
\end{align*}
	Finally, for the case that we restrict to scaling and $x$-translation (but not $y$-translation), we align $P$ and $Q$ along the bottom boundaries of $P_0$ and $Q_0$, such that any scaling and $x$-translation of $P$ will keep the bottom boundaries of $P_0$ and $Q_0$ aligned.

\begin{figure}
	\begin{center}
		\includegraphics[width=0.8\textwidth]{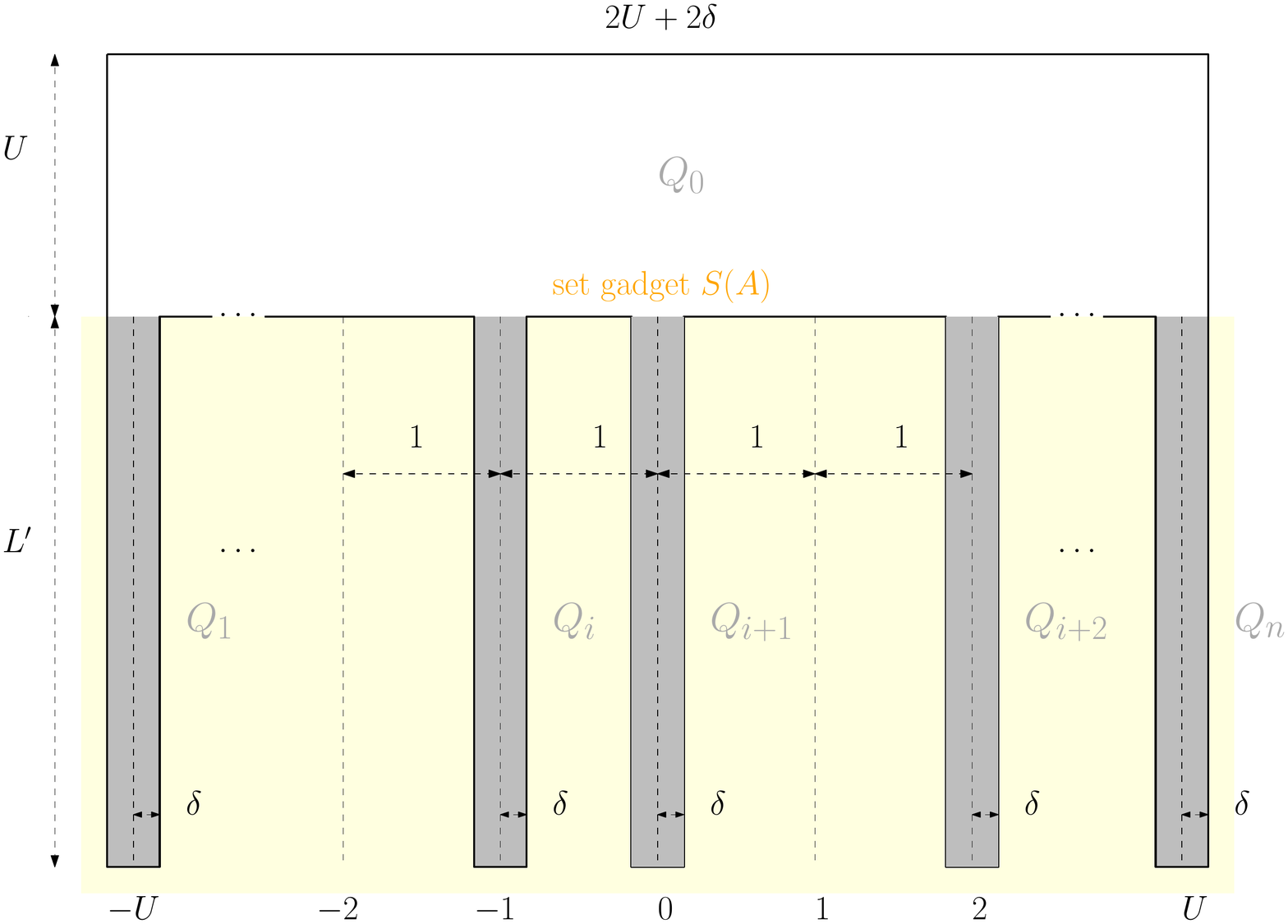}
	\end{center}
	\caption{Polygon $Q$ of the reduction from 3SUM.}
	\label{fig:3SUMq}
\end{figure}

\begin{claim}\label{claim:correctness3SUM}
	There is a scaling factor $\lambda \ge 1$ such that $\lambda P$ can be $x$-translated to be contained in $Q$ if and only if there are $a_1, a_2, a_3 \in A$ such that $a_2-a_1 = a_3-a_2$.
\end{claim}
\begin{proof}
Let $a_1, a_2, a_3 \in A$ such that $a_2-a_1 = a_3-a_2$. We define $\lambda \coloneqq a_2 - a_1$ and observe that $\lambda \le U$ by $a_1,a_2,a_3\in \{-U, \dots, U\}$. We show that $\lambda P$ can be $x$-translated to be contained in $Q$. Specifically, we align the (vertical center of the) leftmost prong $P_1$ of $\lambda P$ with the (vertical center of) $Q_0$'s bottom boundary segment representing the universe element $a_1$. Since $a_1\in A$, $Q$ must have a corresponding prong $Q_{i_1}$ for $a_1$, and by $\lambda \le U$, the scaled prong $P_1$ has length $\lambda L \le UL=L'$ and width $\lambda (2\eps) \le 2U\eps = 2\delta$ and hence is contained in $Q_{i_1}$. By definition of $P$ and $Q$, the center of the second prong $P_2$ of $\lambda P$ is at distance $\lambda$ to the center of $P_1$ and thus aligned with the center of the segment of $Q_0$ representing $a_1 + \lambda = a_1 + (a_2 - a_1) = a_2$. Since $a_2\in A$, again a corresponding prong $Q_{i_2}$ of sufficient length and width exists, which contains the scaled $P_2$. Finally, in the same way, the third prong $P_3$ is aligned with the segment representing $a_2 + \lambda = 2a_2 - a_1 = a_3$, and hence is contained in a corresponding prong $Q_{i_3}$ of $Q$ since $a_3 \in A$.

	Conversely, assume that there is a $\lambda \ge 1$ such that $\lambda P$ can be translated to fit into $Q$ and fix such a translation. Observe that the bottom boundary of the prongs $P_1, P_2, P_3$ of $\lambda P$ must be contained in some prongs $Q_{i_1}, Q_{i_2}, Q_{i_3}$ of $Q$ --- this even holds if we allow $y$-translation, since we chose $L>U$ and thus $\lambda P$ (which is of height $\lambda (L+1) > U$) cannot be fully contained in $Q_0$ (which is of height $U$). For $j \in \{1, \dots, 3\}$, let $s_j$ be the bottom left vertex of prong $P_j$ of $\lambda P$ and let $a_j$ be the element of $A$ represented by prong $Q_{i_j}$. Since the distance between $s_1$ and $s_2$ is $\lambda$ and these vertices are contained in $2\delta$-width prongs $Q_{i_1} ,Q_{i_2}$ whose centers have distance $a_2-a_1$, we conclude that $\lambda \in [(a_2-a_1) - 2\delta, (a_2 - a_1) + 2\delta]$. Thus, since $s_3$ has distance $\lambda$ from $s_2$, it must be contained in any prong whose $x$-dimensions intersect $[(a_2 - 2\delta) + \lambda, (a_2 + 2\delta) + \lambda] \subseteq [2a_2 - a_1 - 4\delta, 2a_2 - a_1 + 4\delta]$. Since $\delta = 1/10 < 1/8$, $Q_{i_3}$ is the \emph{unique} prong intersecting this range, representing the element $a_3 = 2a_2 - a_1$, which proves that $a_2 - a_1 = a_3 - a_2$, as desired.
\end{proof}

We immediately obtain the claimed 3SUM lower bound.
\begin{proof}[Proof of Theorem~\ref{thm:3SUM}]
Observe that the produced polygons are orthogonal and that $P$ and $Q$ have $12$ and $\Oh(n)$ vertices, respectively, and can be constructed in time $\Oh(n)$ for any given \average instance. Thus, for any $\eps > 0$, a $\Oh(n^{2-\eps})$-time algorithm for Polygon Placement would give a $\Oh(n^{2-\eps})$-time algorithm for \average, refuting the 3SUM Hypothesis by 3SUM-hardness of \average~\cite{DudekGS20}.
\end{proof}

In fact, we can even rule out any polynomial approximation factor.
\begin{thm}
	Assuming the 3SUM Hypothesis, there are no $\gamma, \eps > 0$ such that, given orthogonal polygons $P$ and $Q$ of complexity 12 and $q$, we can produce an estimate $\tilde{\lambda}$ in time $\Oh(q^{2-\eps})$ with $\lambda^* \in [\tilde{\lambda}, q^{\gamma} \cdot \tilde{\lambda}]$, where $\lambda^*$ is the largest scaling factor $\lambda$ such that $\lambda P$ can be $x$-translated into $Q$.
\end{thm}

\begin{proof}
	Observe that the above proof already reveals that there exists a placement of $\lambda P$ into $Q$ such that $\lambda P$ intersects more than one prong of $Q$ if and only if the given \average instance has a solution. Thus, if the \average instance has a solution, we obtain $\lambda^* \ge 1$ by Claim~\ref{claim:correctness3SUM}. Otherwise, $\lambda P$ cannot intersect more than one prong of $Q$, and thus must be contained in $Q_0 \cup Q_i$ for some $i$. The only way to do this is to either use a scaling factor $\lambda \le \frac{2\epsilon}{2+\epsilon} < \epsilon$ to fit $\lambda P$ into the width of a single prong, or use a scaling factor $\lambda \le U/L$ such that the prongs' height can be fully contained in $Q_0$. Thus, $\lambda^* \le \max\{ \epsilon, U/L\}$. Finally, the claim follows from observing that in the proof of Claim~\ref{claim:correctness3SUM} we can choose $\epsilon \le 1/(10U)$ arbitrarily small and $L \ge 2U$ arbitrarily large. Thus, we can achieve an arbitrary polynomial gap for $\lambda^*$.
\end{proof}

\section{5SUM Lower Bound for Scaling, Translation, and Rotation}
\label{sec:5SUM}

In this section we present the reduction from 5SUM to the polygon placement problem with scaling, translations, and rotation. We first describe the reduction to then prove its correctness.
We reduce from the following version of 5SUM, which we refer to as 5SUM': Given sets $A_1$, $A_2$, $B_1$, $B_2$, and $C$, each containing $n$ integers in $\{-U, \dots, U\}$ with $U = n^5$, the task is to determine whether there are $a_1 \in A_1, a_2 \in A_2, b_1 \in B_1, b_2 \in B_2$ and $c\in C$ such that 
\[ b_2 - b_1 = a_2 - a_1 + c.\]
To see that this formulation is indeed 5SUM-hard, reduce any 5SUM instance $X_1, \dots, X_5 \subseteq\{-U, \dots, U\}$ to this formulation by defining $A_1 = -X_1, A_2 = X_2, B_1 = X_3, B_2 = -X_4, C=X_5$, where $-X \coloneqq \{-x \mid x\in X\}$ for any set $X$.
Now observe that $x_1 \in X_1, x_2 \in X_2, x_3\in X_3, x_4\in X_4, x_5 \in X_5$ sum up to zero if and only if $(a_1, a_2, b_1, b_2, c) = (-x_1, x_2, x_3, -x_4, x_5) \in A_1 \times A_2 \times B_1 \times B_2 \times C$ yields a solution to the corresponding 5SUM' instance.


\subsection*{Reduction.}

\begin{figure}
\centering
\includegraphics[width=.62\textwidth]{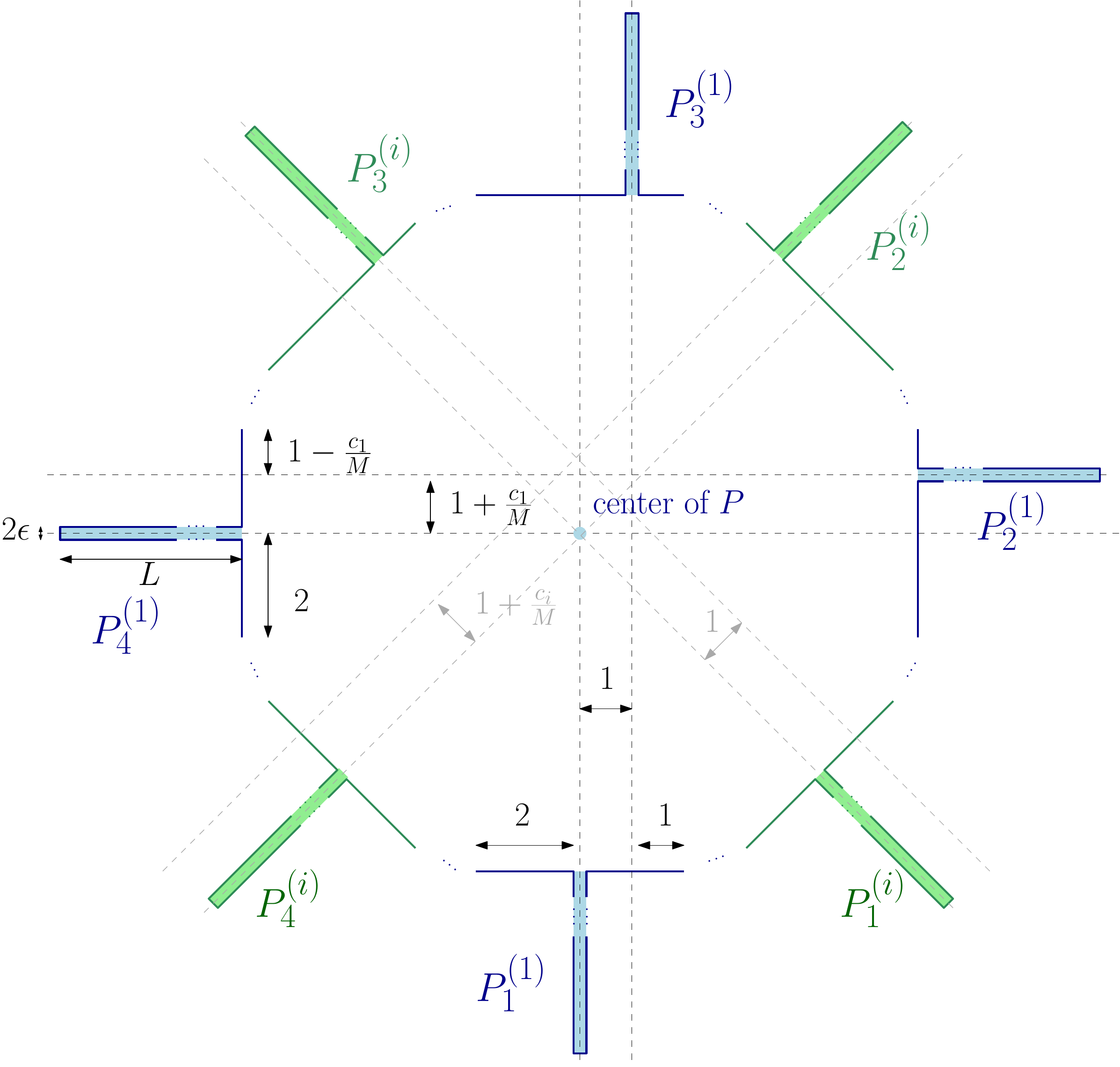}
\caption{Polygon $P$ of reduction for scaling, translation, and rotation.}
\label{fig:5SUMp}

\vspace*{\floatsep}

\includegraphics[width=.62\textwidth]{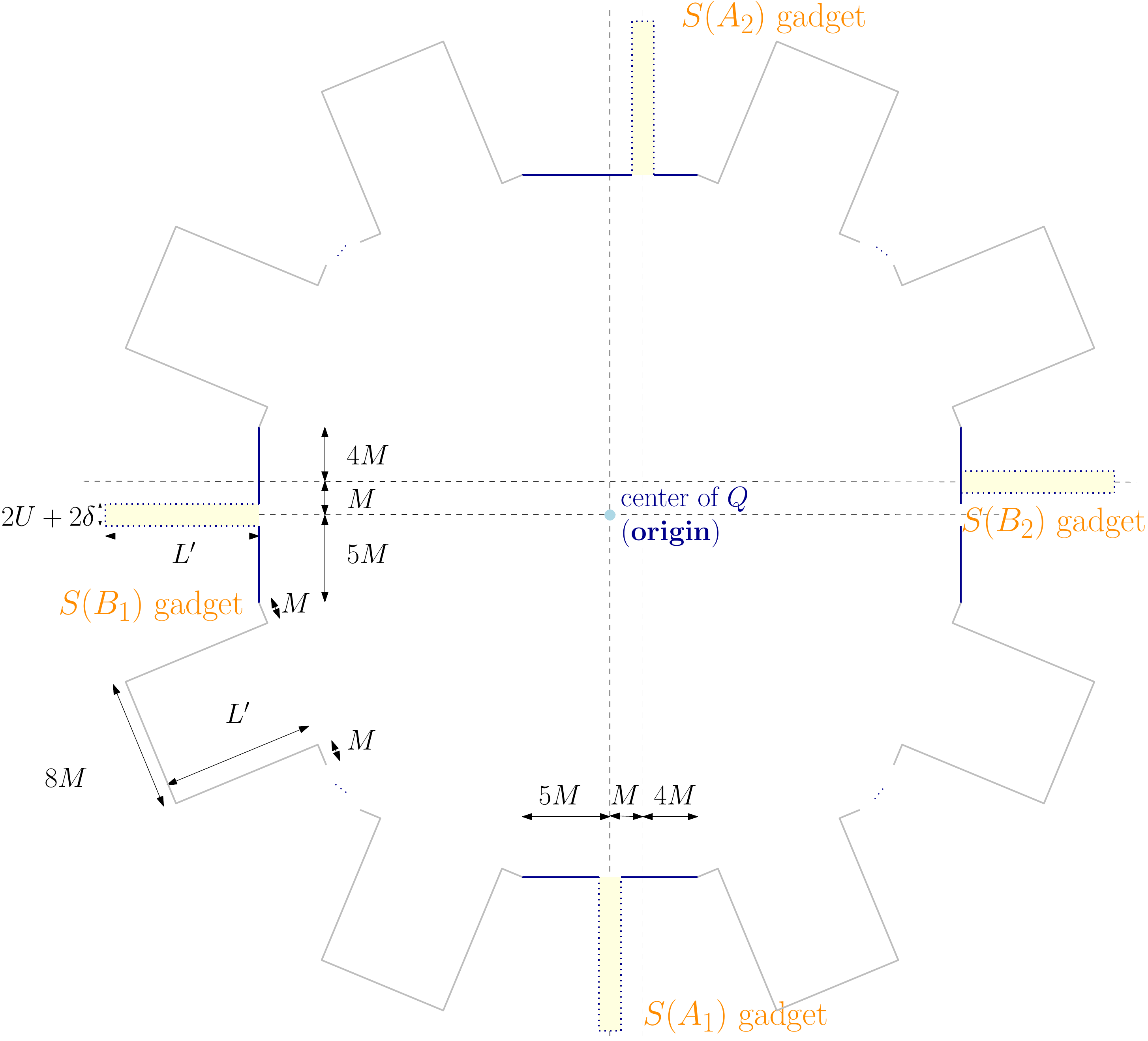}
\caption{Polygon $Q$ of reduction for scaling, translation, and rotation.}
\label{fig:5SUMq}
\end{figure}


We now describe the reduction. 
Let $M$, $L$, $L'$, $\eps$, and $\delta$ be values that we specify later.
Given a 5SUM' instance, we define the polygons $P$ and $Q$ as illustrated in Figures~\ref{fig:5SUMp}~and~\ref{fig:5SUMq}. Both polygons are $4n$-gons with $\Oh(n)$ attached prongs. Specifically, let $P_0$ be the regular $4n$-gon with side length $4$. We partition the $4n$ sides of $P_0$ into $n$ groups, where the $i$-th group consists of two opposing segments $s^{(i)}_1, s^{(i)}_3$ together with the pair of opposing segments $s^{(i)}_2, s^{(i)}_4$ that are orthogonal to $s^{(i)}_1, s^{(i)}_3$. Specifically, $s_1^{(1)}, s_2^{(1)}, s_3^{(1)},$ and $s_4^{(1)}$ denote the bottom horizontal, right vertical, top horizontal, and left vertical segments, respectively. The groups $2,\dots, n$ are obtained analogously through rotation in clockwise direction. At each segment $s^{(i)}_j$, we attach a prong $P^{(i)}_j$ --- a rectangle of length $L$ and width $2\eps$ --- at a carefully chosen position. In particular, the distance between the placements of prongs $P_2^{(i)}$ and $P_4^{(i)}$ will encode the $i$-th element $c_i$ of $C$.

Formally, let us describe the placement of prong $P_j^{(i)}$ relative to the center point  of $s_j^{(i)}$, according to traversing the segments in anti-clockwise order (refer to Figure~\ref{fig:5SUMp} for details). We place the center segments of $P_3^{(i)}$, $P_4^{(i)}$, and $P_1^{(i)}$ at a distance of $-1$, $0$, and $0$, respectively, from the center point of the corresponding segment.\andre{Don't know how to reformulate, but it is hard to understand.} The center segment of $P_2^{(i)}$ is placed at a distance of $1+\frac{c_i}{M}$, where $c_i$ is the $i$-th element in $C$. Crucially, if we rotate the polygon such that $P_{1}^{(i)}$ is a vertical prong, then the horizontal centers of $P_1^{(i)}$ and $P_3^{(i)}$ have a horizontal distance of $1$, and the vertical centers of prongs $P_2^{(i)}$ and $P_4^{(i)}$ have a vertical distance of $1+\frac{c_i}{M}$, see Figure~\ref{fig:5SUMp}.

The polygon $Q$ is defined in a similar way, see Figure~\ref{fig:5SUMq}: It is defined as a regular $4n$-gon $Q_0$ of side length $10M$ where we attach to each side either a \emph{set gadget} (consisting of $n$ prongs of width $2\delta$ and length $L'$) or a \emph{trivial gadget}, i.e., a large, rotated $8M\times L'$ rectangle which can fit any reasonably placed prong of $P$. The idea is that any large copy of $P$ must be rotated such that the prongs $P_1^{(i)}, P_2^{(i)}, P_3^{(i)}, P_4^{(i)}$ must fit into the set gadgets $S(A_1), S(B_1), S(A_2), S(B_2)$ as defined in Section~\ref{sec:3SUM}, Figure~\ref{fig:3SUMq}. Specifically, at the horizontal and vertical segments of $Q_0$ we place the set gadgets $S(A_1), S(B_1), S(A_2), S(B_2)$ relative to the center points  in the following way (according to an anticlockwise traversal of the segments): we place $S(B_1)$, $S(A_1), S(B_2), S(A_2)$ such that they are centered around the distance $0, 0, M, -M$ from the start point of the corresponding segment.\andre{Don't know how to reformulate, but it is hard to understand.} In this way, the horizontal centers of $S(A_1)$, $S(A_2)$ have a horizontal distance of $M$ and the vertical centers of $S(B_1)$, $S(B_2)$ have a vertical distance of $M$. For all remaining segments, we place trivial gadgets centered around the center point of the corresponding gadget.

We choose the parameters as follows:\footnote{These parameters are not optimized to have the smallest possible bit length. In any case, all coordinates in the description of the produced instance will remain polynomially bounded.}
\begin{align*}
	M &= U^2 n,&  L&= 50 n,& \eps &= \frac{1}{800Mn}, \\
	& & L' &= 2ML,& \delta &= \frac{1}{400}.
\end{align*}

\subsection*{Correctness.}
We show that this construction gives a Polygon Placement instance equivalent to the given 5SUM' instance. We split the proof into two lemmas, one for each direction.
For the next two proofs, recall that for any $x,y,\eps \in \mathbb{R}$, we use $x \in y + [-\eps, \eps]$ to mean $x \in [y - \eps, y + \eps]$.

\begin{lem}\label{lem:5SUMhardness1}
If there are $a_1\in A_1, a_2\in A_2, b_1 \in B_1, b_2 \in B_2$ and $c\in C$ such that $b_2-b_1 = a_2 - a_1 + c$, then there is a scaling factor $\lambda \ge M-2U$ such that $\lambda P$ can be rotated and translated to be contained in $Q$.
\end{lem}
\begin{proof}
	Let $a_1\in A_1, a_2 \in A_2, b_1 \in B_1, b_2 \in B_2$ and $i^*\in \{1,\dots,n\}$ be such that $b_2 - b_1 = a_2 - a_1 + c_{i^*}$ where $c_{i^*}$ is the $i^*$-th element of $C$. We claim that we can place $\lambda P$ with scaling factor $\lambda = M + a_2 - a_1$ by translating its origin by $\tau = (a_1, b_1)$ and rotating the polygon by an angle of $- 2\pi \frac{i^* - 1}{4n}$ around its origin. Note that this implies $\lambda \in [M - 2U, M + 2U]$.

	We first observe that under this scaling, rotation and translation, $P_0$ fits into $Q_0$: We can conservatively approximate $P_0$ and $Q_0$ by their circumscribed circle and inscribed circle, respectively, whose radii are the circumradius $R_{P_0}$ and the apothem $a_{Q_0}$, respectively. Since their side lengths are $4\lambda$ and $10M$, respectively, we notice by Fact~\ref{fact:polygons} and $\norm{\tau}_2 \le \sqrt{2}U$ that 
	\[R_{P_0} + \norm{\tau}_2 \le 4\lambda n + \sqrt{2}U \le 4(M+2U)n + \sqrt{2} U \le 5Mn \le a_{Q_0},\]  
	where the second-to-last inequality follows from $M=U^2n \ge 8Un + \sqrt{2}U$ for sufficiently large $n$, as $U=n^5$. Thus, when shifting an arbitrarily rotated $\lambda P$ by $\tau$, $P_0$ will remain in the inscribed circle of $Q_0$.

	It remains to show that each prong $P_j^{(i)}$ fits into $Q$ under this placement. We first consider the most important case of $i=i^*$. Note that by the chosen rotation, we have that $P_{1}^{(i^*)}, P_{2}^{(i^*)}, P_{3}^{(i^*)}$ and $P_{4}^{(i^*)}$ are the bottom horizontal, right vertical, top horizontal and left vertical prongs of $\lambda P$. We claim that the part of $P_1^{(i^*)}$ that leaves $Q_0$ is contained in the prong in $S(A_1)$ representing $a_1$. Indeed, note that the center segment of $P_1^{(i^*)}$ intersects the bottom segment of $Q_0$ at an $\x$-coordinate of $\tau_\x = a_1$. 
As $\lambda \le M+2U \le 2M$, we have that $P_{1}^{(i^*)}$ is a prong of width $2\lambda \eps \le 4M\eps \le \frac{1}{200n} \le 2\delta$ and length $\lambda L \le 2ML = L'$. Hence, this part fits fully into the width-$2\delta$, length-$L'$ prong of $S(A_1)$ centered at $\x$-coordinate $a_1$. 
	Likewise, the center segment of prong $P_4^{(i^*)}$ intersects the left vertical segment at a $\y$-coordinate of $\tau_\y = b_1$, and thus the prong $P_4^{(i^*)}$ is fully contained in the union of $Q_0$ and the prong representing $b_1$ in $S(B_1)$ (attached at $\y$-coordinates $b_1 + [-\delta, \delta]$).
	For prong $P_3^{(i^*)}$, we observe that its center segment intersects the top horizontal segment of $Q_0$ at $\x$-coordinate $\tau_\x + \lambda = a_1 + (M + a_2 - a_1) = M + a_2$. Thus, $P_3^{(i^*)}$ is fully contained in the union of $Q_0$ and the prong representing $a_2$ in $S(A_2)$ (attached at $\x$-coordinates $M+a_2 + [-\delta, \delta]$). Finally, the center segment of prong $P_2^{(i^*)}$ intersects the right vertical segment of $Q_0$ at a $\y$-coordinate of
\[
	\tau_y + \lambda\left(1+\frac{c_{i^*}}{M}\right) = b_1 + (M+a_2-a_1)\left(1+\frac{c_{i^*}}{M}\right) = b_1 + M + a_2 - a_1 + c_{i^*} + \frac{(a_2-a_1)c_{i^*}}{M}.
\]
	Using that $b_1 + a_2 - a_1 + c_{i^*} = b_2$ by assumption, and that $\left|\frac{(a_2-a_1)c_{i^*}}{M}\right| \le \frac{2U^2}{M} = \frac{2}{n}$, we obtain that the prong $P_2^{(i^*)}$ intersects the right vertical segment of $Q_0$ at a $\y$-coordinate of $M+b_2 + [-(\lambda \eps + \frac{2}{n}), \lambda \eps + \frac{2}{n}]$. By noting that $\lambda \eps + \frac{2}{n} \le 2M\eps + \frac{2}{n} = \frac{1}{400n} + \frac{2}{n} \le \frac{1}{400} = \delta$ for sufficiently large $n$, this range is fully contained in the range $M+b_2 + [-\delta, \delta]$ given by the prong representing $b_2$ in $S(B_2)$, which thus contains the relevant part of $P_2^{(i^*)}$. We conclude that each prong of the $i^*$-th group fits into $Q_0$.

	It remains to argue that each prong $P_j^{(i)}$ with $i\ne i^*$ fits into $Q_0$. This claim is essentially trivial: Imagine a rotation of coordinate axes such that $P_j^{(i)}$ is a vertical prong. It is easy to see that $P_j^{(i)}$ intersects the bottom segment of $Q_0$ at $\x$-coordinates in $[-2M, 2M]$. The corresponding $4M\times L'$ rectangle below $Q_0$ is contained in the trivial gadget, and thus can contain the part of $P_j^{(i)}$ leaving $Q_0$. This concludes the lemma.
\end{proof}

We now show the converse of Lemma \ref{lem:5SUMhardness1}.
\begin{lem}\label{lem:5SUMhardness2}
If there is a scaling factor $\lambda \ge M-2U$ such that $\lambda P$ can be rotated and translated to be contained in $Q$, then there are $a_1\in A_1, a_2\in A_2, b_1 \in B_1, b_2 \in B_2$ and $c\in C$ such that $b_2-b_1 = a_2 - a_1 + c$.
\end{lem}
\begin{proof}
Assume that there is a scaling factor $\lambda \ge M-2U$ such that $\lambda P$ can be placed into $Q$ (under rotation and translation) and fix such a placement of $\lambda P$. Specifically, let $\lambda \ge M-2U$ be the scaling factor, $\tau \in \mathbb{R}^2$ denote the translation of the center of $\lambda P$, and let $\alpha \in [0, 2\pi)$ denote the rotation angle of $\lambda P$ around its center (in counter clockwise direction). In the remainder of this proof, we show that the following sequence of statements holds:
\begin{enumerate}[label=(\arabic*)]
	\item The rotation angle $\alpha$ is in $- 2\pi \cdot \frac{i}{4n}  + [-\eps_\rot,\eps_\rot]$ for some $i\in \{0, \dots, 4n-1\}$ and $\eps_\rot \coloneqq \frac{\delta}{M n}$. Furthermore, there is some prong $P^{(i^*)}_{j^*}$ that intersects the set gadget $S(A_1)$.
\item We have $j^* \in \{1,3\}$ and there are $a_1\in A_1, a_2 \in A_2$ such that the $x$-translation $\tau_\x$ is within $a_1 + [-\eps_\trans, \eps_\trans]$ (if $j^* = 1$) or within $M + a_2 + [-\eps_\trans, \eps_\trans]$ (if $j^*=3$) and the scaling factor $\lambda$ is within $M + a_2 - a_1 + [-\eps_\scale, \eps_\scale]$, where $\eps_\trans \coloneqq 21\delta$ and $\eps_\scale \coloneqq 43\delta$. 
\item If $j^* = 1$, then there is some $b_1\in B_1$ such that the $y$-translation $\tau_\y$ is within $b_1 + [-\eps_\trans, \eps_\trans]$. If $j^* = 3$, there is some $b_2 \in B_2$ such that the $y$-translation $\tau_\y$ is within $M + b_2 + [-\eps_\trans, \eps_\trans]$.
\item Finally, there is a corresponding $b_2\in B_2$ (if $j^* = 1$) or $b_1 \in B_1$ (if $j^* = 3$) such that $b_2 - b_1 = a_2 - a_1 + c_{i*}$.
\end{enumerate}
After showing the above statements, we can then deduce that any placement of $\lambda P$ in $Q$ where $\lambda \ge M-U$ uniquely defines $a_1\in A_1, a_2 \in A_2, b_1\in B_1, b_2\in B_2, c\in C$ such that $b_2 - b_1 = a_2 - a_1 + c$.

We first collect general useful facts: By $M=U^2 n$, we have $M-2U = M(1-2/(Un))\ge 0.9M$ for sufficiently large $n$, which we will use as a lower bound for $\lambda$. A trivial upper bound is given by $\lambda \le 3M$, as otherwise already the scaled $P_0$ (of side length more than $12M$) cannot fit into $Q$ (of side length $10M$). By Fact~\ref{fact:polygons}, the regular $4n$-gon $Q_0$ with side length $10M$ has a circumradius $R \le 10Mn$. Thus, the diameter of $Q_0$ is bounded by $2R \le 20Mn$, which we will frequently exploit. \marvin{define center segment and line somewhere}

To prove (1), note that since $\lambda L \ge  0.9 ML = 45 Mn > 20Mn \ge 2R$, each prong of $\lambda P$ (each of length $\lambda L$) cannot be fully contained in $Q_0$, but must intersect either a set gadget or a trivial gadget. In fact, by the structure of $P$ and $Q$, every prong $P_{j}^{(i)}$ of $\lambda P$ must intersect a \emph{unique} (set or trivial) gadget of $Q$\marvin{more detailed argument?}. In particular, let $P_{j^*}^{(i^*)}$ denote the prong of $\lambda P$ intersecting $S(A_1)$. We will show that the rotation angle is very close to $- 2\pi \cdot \frac{(j^* - 1)n + i^* - 1}{4n}$: Let $\beta \coloneqq  \alpha + 2\pi \cdot \frac{(j^* - 1)n + i^* - 1}{4n}$ be the difference of $\alpha$ and this angle. Observe that the center segment of $P_{j^*}^{(i^*)}$ has an intersection of length at least $\lambda L - 2R \ge 45Mn - 20Mn = 25Mn$ inside a prong of $S(A_1)$ and cannot traverse a horizontal distance of more than $2\delta$ (as any prong of $S(A_1)$ has width $2\delta$). Since under rotation angle $\alpha$, $P_{j^*}^{(i^*)}$ traverses a horizontal distance of at least $25Mn \cdot |\sin(\beta)|$, we must have that $|\sin(\beta)| \le \frac{2\delta}{25Mn}$, and thus by Fact~\ref{fact:trig-approx}, $|\beta| \le (\pi/2)\cdot |\sin(\beta)| \le 2 \cdot \frac{2\delta}{25Mn} \le \frac{\delta}{Mn}$, concluding the proof of (1).

To prove (2), let $P_{\mathrm{top}}$ denote the prong opposite of $P_{j^*}^{(i^*)}$. Note that if $j^* \in \{2, 4\}$ then $P_{\mathrm{top}}$  is to the left of $S(A_1)$, and hence cannot intersect $S(A_2)$ (the unique gadget that $P_{\mathrm{top}}$ must intersect), as $P_{\mathrm{top}}$ is located to the right of $S(A_1)$. Thus, we must have $j^*\in \{1, 3\}$. If $j^*=1$, then for $P_{j^*}^{(i^*)}$ to intersect some prong of $S(A_1)$, the translation $\tau_x$ must be such that the center segment of $P_{j^*}^{(i^*)}$ intersects the bottom horizontal segment of $Q_0$ at some prong of $S(A_1)$, i.e., at an $\x$-coordinate in $a_1 + [-\delta, \delta]$ for some $a_1 \in A_1$. Note that the center line of $P_{j^*}^{(i^*)}$ runs through the origin of $P$ and that any segment of $Q_0$ is at distance at most $2R$ from the translated center of $\lambda P$. Thus, by the Line Rotation Lemma~\ref{lem:line-rotation} (setting $\Delta = 0, Y = 2R, \eps = \eps_\rot$), the center segment intersects $Q_0$'s bottom segment at an $\x$-coordinate in $\tau_\x + [-2R \eps_\rot,2R\eps_\rot]$. For $a_1 + [-\delta, \delta]$ and $\tau_x + [-2R \eps_\rot, 2R\eps_\rot]$ to intersect, $\tau_\x$  must be in $a_1 + [- (\delta + 2R \eps_\rot), \delta + 2R \eps_\rot]$. By observing that $2R \eps_\rot \le \frac{20Mn\delta}{Mn} = 20\delta$, we obtain that $\tau_\x$ is indeed in $a_1 + [-21\delta, 21\delta]$. For the scaling factor, recall that the center segment of the opposite prong $P_\mathrm{top}$ must intersect some prong of $S(A_2)$, i.e., at some $\x$-coordinate in $M + a_2 + [-\delta, \delta]$ for some $a_2 \in A_2$. By the Line Rotation Lemma~\ref{lem:line-rotation} (setting $\Delta = \lambda, Y = 2R, \eps = \eps_\rot$), the center segment intersects the top horizontal segment of $Q_0$ at an $\x$-coordinate in $\tau_\x + [\lambda - 2R\eps_\rot, \lambda(1+\eps_\rot^2) + 2R\eps_\rot]$. Since $\tau_\x$ is in $a_1 + [-21\delta, 21\delta]$ and $\lambda\eps_\rot^2 \le 3M\eps_\rot^2 = \frac{3\delta^2}{Mn^2} \le \delta$, we obtain that the intersection has an $\x$-coordinate in $a_1 + \lambda + [- (2R\eps_\rot + 21\delta), 2R\eps_\rot + 22\delta]$. Since $2R\eps_\rot \le 20 \delta$, and $a_1 + \lambda + [-42\delta, 42\delta]$ intersects $M + a_2 + [-\delta, \delta]$, we obtain that $\lambda$ is in $M+ a_2 - a_1 + [-43\delta, 43\delta]$.

If $j^*=3$, i.e., $P$ is as before, but rotated by $\pi$, the calculations are analogous, but choose first $a_2$, then $a_1$: For $P_{\mathrm{top}}$ to intersect $S(A_2)$, the $x$-translation $\tau_\x$ must be in $M + a_2 + [-21\delta, 21\delta]$ for some $a_2 \in A_2$. Correspondingly, for $P_{j^*}^{(i^*)}$ to intersect $S(A_1)$, the scaling factor $\lambda$ must be in $M + a_2 - a_1 + [-43\delta, 43\delta]$ for some $a_1\in A_1$.

To prove (3), let $P_\mathrm{left}$ ($P_\mathrm{right}$) denote the prongs in the $i^*$-th group following (preceding) $P_{j^*}^{(i^*)}$ in clockwise order, i.e., $P_\mathrm{left}, P_\mathrm{right}$ must intersect $S(B_1), S(B_2)$, respectively. If $j^*=1$, then the intersection of the center line of $P_{j^*}^{(i^*)}$ with the left vertical segment of $Q_0$ is, by the Line Rotation Lemma~\ref{lem:line-rotation} (with $\Delta = 0, Y=2R, \eps = \eps_\rot$, appropriately rotated), at a $\y$-coordinate of $\tau_\y + [-2R\eps_\rot, 2R\eps_\rot]$. For this to intersect $b_1 + [-\delta, \delta]$ for some $b_1 \in B_1$, we must have that $\tau_y$ is within $b_1 + [-(\delta + 2R\eps_\rot), \delta + 2R\eps_\rot]$. Since, $2R\eps_\rot \le 20 \delta$, we conclude that $\tau_\y$ is within $b_2 + [-21\delta, 21\delta]$. If $j^* = 3$, then the center line of $P_\mathrm{right}$ intersects the right vertical segment of $Q_0$ at $\tau_y + [-2R\eps_\rot, 2R\eps_\rot]$ by the Line Rotation Lemma~\ref{lem:line-rotation}. For this to intersect some prong of $S(B_2)$, it must intersect $M + b_2 + [-\delta, \delta]$ for some $b_2\in B_2$. Analogously to before, we conclude that $\tau_\y$ is within $M + b_2 + [-21\delta, 21\delta]$.

Finally, we can prove (4): If $j^* = 1$, then recall that $\tau_\x, \tau_\y$ and $\lambda$ are in $a_1 + [-\eps_\trans, \eps_\trans], b_1 + [-\eps_\trans, \eps_\trans]$ and $M+a_2-a_1 + [-\eps_\scale, \eps_\scale]$ and that the prong $P_{\mathrm{right}}$ must intersect some prong in $S(B_2)$. The intersection of the center segment of $P_\mathrm{right}$ with the right vertical segment of $Q_0$ is, by the Line Rotation Lemma~\ref{lem:line-rotation} (with $\Delta = \lambda(1+\frac{c_{i^*}}{M}), Y=2R$, appropriately rotated) at a $\y$-coordinate within $\tau_y + [\lambda(1+\frac{c_{i^*}}{M}) - 2R\eps_\rot, \lambda(1+\frac{c_{i^*}}{M})(1+\eps_\rot^2) + 2R\eps_\rot]$. Noting that $2R\eps_\rot \le 20\delta$ and $\lambda(1+\frac{c_{i^*}}{M})\eps_\rot^2 \le 3M(1+\frac{U}{M})\eps_\rot^2 \le 6M \eps_\rot^2 = \frac{6\delta^2}{Mn^2}\le \delta$ for sufficiently large~$n$, we conclude that the intersection is at some $\y$-coordinate in $\tau_\y + \lambda(1+\frac{c_{i^*}}{M}) + [-21\delta, 21\delta]$. Observe that by $\lambda$ being in $M+a_2-a_1 + [-43\delta, 43\delta]$ and $1+\frac{c_{i^*}}{M} \le 2$, we have that 
$\lambda\left(1+\frac{c_{i^*}}{M}\right)$ is within
\[
M+a_2-a_1 + c_{i^*} + \frac{c_{i^*}(a_2-a_1)}{M} +  [-86\delta, 86\delta]
\]
and hence within
\[
M+a_2-a_1 + c_{i^*} + [-(86\delta + \frac{2U^2}{M}), 86\delta + \frac{2U^2}{M}].
\]
		By choice of $M$, we have $\frac{2U^2}{M} = \frac{2}{n} \le \delta$ for sufficiently large $n$. Thus, the intersection of the center line of prong $P_\mathrm{right}$ with the right vertical segment of $Q_0$ is at $\tau_\y + M+a_2-a_1 + c_{i^*} + [-(87\delta + 21 \delta), 87\delta + 21 \delta]$. Using that $\tau_\y$ is within $b_1 + [-21\delta, 21\delta]$, we conclude that the intersection is at $M + b_1 + a_2-a_1  +c_{i*} + [-129\delta, 129\delta]$. For this intersection to be contained in $M+b_2 + [-\delta, \delta]$ for some $b_2 \in B_2$, we must have that 
		\[|(M+b_2)  - (M+ b_1 + a_2-a_1 + c_{i^*})| \le 130 \delta.\]
		Since $130 \delta < 1/2$ and $|(M+ b_2) - (M+b_1 + a_2-a_1 + c_{i^*})| = |b_2 - (b_1 + a_2 - a_1 + c_{i^*})|$ is integral, we conclude that the unique possibility is that there exists $b_2\in B_2$ which satisfies $b_2 = b_1 + a_2 - a_1 + c_{i^*}$. 

Analogous calculations show the statement for the case of $j^* = 3$: We already know that $\tau_\x, \tau_\y$ and $\lambda$ are in $M + a_1 + [-\eps_\trans, \eps_\trans], M + b_2 + [-\eps_\trans, \eps_\trans]$ and $M+a_2-a_1 + [-\eps_\scale, \eps_\scale]$. Analogously to above, the intersection of the center segment of $P_\mathrm{left}$ with the left vertical segment of $Q_0$ can be shown to be within $b_2 - (a_2 - a_1 + c_{i^*}) + [-129\delta, 129\delta]$. For this to intersect some prong in $S(B_1)$, i.e., $b_1 + [-\delta, \delta]$ for some $b_1 \in B_1$, we must have $|b_1 - (b_2 - a_2 + a_1 - c_{i^*}| \le 82\delta$, which by $130\delta < \frac{1}{2}$ and integrality of the left hand side proves $b_1 = b_2 - a_2 + a_1 - c_{i^*}$ is in $B_1$.
\end{proof}

Finally, we can prove the result of this section.
\begin{proof}[Proof of Theorem~\ref{thm:5SUM}]
	The equivalence between the polygon containment and the 5SUM' instance is shown in Lemmas~\ref{lem:5SUMhardness1}~and~\ref{lem:5SUMhardness2}. It only remains to argue that the claimed lower bound follows. Observe that $P$ is a $4n$-gon with attached rectangles at all sides and thus consists of $\Oh(n)$ vertices. The polygon $Q$ also is a $4n$-gon with constant complexity rectangles attached at all sides, however, it additionally has gadgets on the bottom, right, top, and left. Each of these four gadgets has complexity $\Oh(n)$ and thus $Q$ has complexity $\Oh(n)$. Thus, since it is straightforward to compute the $O(n)$-vertex polygons $P$ and $Q$ in linear time, the lower bound follows.
\end{proof}

\section{4SUM Lower Bound for Scaling and Arbitrary Translation}
\label{sec:4SUM}

After having established the 5SUM-hardness for scaling, translation, and rotation (Theorem~\ref{thm:5SUM}), it is immediate how to obtain Theorem~\ref{thm:4SUM}.

\begin{proof}[Proof of Theorem~\ref{thm:4SUM}]
	The proof follows from simplifying the reduction in the proof of Theorem~\ref{thm:5SUM} in the previous section by viewing $C$ as a singleton set containing $0$, and using a single group instead of $n$ groups. It is straightforward to adapt the proof of correctness and to observe that $P$ simplifies to a square with $4$ attached prongs, resulting in a constant-sized polygon.
\end{proof}

\section{Algorithm for Scaling and Translation}
\label{sec:algo}

We will prove Theorem~\ref{thm:algo} by a reduction to the following offline dynamic problem that we refer to as \textsc{Offline Dynamic Rectangle Cover}: Let $\mathcal{S}$ be a set of $n$ rectangles in $[1,N]\times [1,N]$ with $N=\Oh(n)$. The input is a sequence of $U$ updates $u_1,\dots, u_U$, where an update comes in two flavors:
\begin{itemize}
	\item \textbf{delete} a rectangle in $\mathcal{S}$, or
	\item \textbf{add} a rectangle with integral coordinates in $[1,N]\times [1,N]$ to $\mathcal{S}$,
\end{itemize}
such that each update maintains $|\mathcal{S}| \le n$. The task is to determine the first update $u_i$ after which $\mathcal{S}$ does not cover $[1,N]\times [1,N]$. Our main algorithmic contribution is to prove the following reduction.
Recall that, given two polygons $P$ and $Q$, $\lambda^*$ is defined as the largest scaling factor $\lambda$ such that $\lambda P$ can be placed into $Q$ under translations.

\begin{lem}\label{lemma:reduction}
	Let $T(n,U)$ denote the optimal running time for solving  \textsc{Offline Dynamic Rectangle Cover}. Then, given any orthogonal simple polygons $P$ and $Q$ with $p$ and $q$ vertices, respectively, we can compute $\lambda^*$ in time $\Oh( (pq)^2 \log(pq) + T(pq, (pq)^2))$.
\end{lem}

Thus, we first show how Lemma \ref{lemma:reduction} implies Theorem \ref{thm:algo}, to then dedicate the remainder of this section to proving Lemma \ref{lemma:reduction}. To this end, we can plug in an offline dynamic algorithm due to Overmars and Yap~\cite{OvermarsY91}, and as improved by~\cite{Chan10}.
\begin{proof}[Proof of Theorem~\ref{thm:algo}]
	Building on Overmars and Yap~\cite{OvermarsY91}, Chan~\cite{Chan10} shows how to maintain the area (in a prespecified box) of the union of at most $n$ rectangles under $\Oh(n)$ pregiven\footnote{Actually, these algorithms only require advance knowledge of the vertices of the rectangles.} insertions or deletions of rectangles (never exceeding $n$ rectangles), in amortized time $\Oh(\sqrt{n}2^{\Oh(\log^* n)})$ per update.

	Thus, to answer $U\ge n$ updates in \textsc{Offline Dynamic Rectangle Cover}, we can divide all updates into $\Oh(U/n)$ batches of at most $n$ updates. For each batch, we use Chan's data structure as follows: we use at most $n$ insertions to create the initial state at the start of the batch. We then perform each update in the batch, stopping if the area becomes strictly less than $(N-1)^2$ and returning this update. Observe that the at most $2n$ updates to Chan's data structure for the current batch are determined in advance. Thus, the total time to perform these updates is $\Oh(n \sqrt{n} 2^{\Oh(\log^* n)})$. Doing this for all $\Oh(U/n)$ batches results in a total time of $T(n,U) = \Oh(\frac{U}{n} n \sqrt{n} 2^{\Oh(\log^* n)}) = \Oh(U \sqrt{n} 2^{\Oh(\log^* n)})$.

	Plugging in $n=\Oh(pq)$ and $U=\Oh((pq)^2)$, we obtain a running time of $\Oh((pq)^{2.5} 2^{\Oh(\log^* pq)})$ for determining $\lambda^*$ for the orthogonal polygons $P,Q$ with $p$ and $q$ vertices, respectively. Given $\lambda^*$, we can also find a feasible translation for $\lambda^* P$ using the algorithm for the fixed-size case~\cite{Barrera96algo} in time $\Oh(pq \log(pq))$, which does not affect the overall running time.
\end{proof}

In the remainder of this section, we prove Lemma~\ref{lemma:reduction}.
Let $P$ and $Q$ be centered at the origin (i.e., their bounding boxes have the origin as their center) and let $B = [x^B_0,x^B_1]\times [y^B_0,y^B_1]$ denote the bounding box of $Q$.
Cover $P$ by rectangles $P_1\dots, P_{p'}$ with $p' = \Oh(p)$ and $\mathbb{R}^2 \setminus Q$ by rectangles $Q_1,\dots, Q_{q'}$ with $q' = \Oh(q)$, where $\mathbb{R}^2 \setminus Q$ denotes the set of points in the plane that do not lie inside $Q$. (To see that this is possible, repeatedly chop off concavities of $P$. Similarly for $Q$, use four rectangles to cover $\mathbb{R}^2\setminus B$ and repeatedly chop off concavities of each connected part in $B \setminus Q$.) For $i\in [p'],j\in [q']$, let $R_{i,j}(\lambda)$ denote the set of translations of $\lambda P$ such that $P_i$ (under this scaling and translation) intersects $Q_j$ in its interior, more formally
\[
	R_{i,j}(\lambda) \coloneqq \{ \tau \in \RR^2 | (\lambda P_i + \tau) \cap (Q_j \setminus \partial Q_j) \neq \emptyset \},
\]
where $\partial Q_j$ is the boundary of $Q_j$ and $\lambda P_i$ is scaled with the center of $P$ as reference point.
\begin{prop}[{\cite{AvnaimB89},\cite[Proposition 1]{Barrera96algo}}]
	For any $\lambda$, $\lambda P$ can  be translated to be contained in $Q$ if and only if $\bigcup_{i,j} R_{i,j}(\lambda)$ does not cover $B$.
\end{prop}
\begin{proof}
	Note that if $\tau\in \RR^2$ is a translation such that $\lambda P$ fits into $Q$, $\tau$ must be a translation in $B$ (as the center of $P$ must be contained in $Q$) and no $P_i$ may intersect some $Q_j$ in its interior. Thus, any such $\tau$ must be in $B\setminus \bigcup_{i,j} R_{i,j}(\lambda)$. Conversely, if we translate $\lambda P$ by some $\tau \in B\setminus \bigcup_{i,j} R_{i,j}(\lambda)$, it must be fully contained in $Q$, as it has no intersection with $\mathbb{R}^2\setminus Q$.
\end{proof}

Note that each $R_{i,j}$ is of the form
\[
	R_{i,j}(\lambda) = (a_{i,j}(\lambda), b_{i,j}(\lambda))\times (c_{i,j}(\lambda), d_{i,j}(\lambda)),
\]
where each $a_{i,j},b_{i,j}, c_{i,j}, d_{i,j}$ is a linear function in $\lambda$. Crucially, while the coordinates of virtually all $R_{i,j}(\lambda)$'s change with each change in $\lambda$, we will work in \emph{rank space} and only consider changes in the \emph{combinatorial structure} of the arrangement of the $R_{i,j}(\lambda)$'s. To this end, define 
\begin{align*} 
	\mathcal{X} &\coloneqq \{ a_{i,j}(\lambda) \mid (i,j)\in [p']\times[q']\} \cup \{ b_{i,j}(\lambda) \mid (i,j)\in [p']\times[q']\} \cup \{x_0^B, x_1^B\}\\
	\mathcal{Y} &\coloneqq \{ c_{i,j}(\lambda) \mid (i,j)\in [p']\times[q']\} \cup \{ d_{i,j}(\lambda) \mid (i,j)\in [p']\times[q']\} \cup \{y_0^B, y_1^B\}
\end{align*}
as the sets of relevant $\x$- and $\y$-coordinates, respectively. We will exploit that the sorted orders of $\mathcal{X}$ and $\mathcal{Y}$ change only at $\Oh((pq)^2)$ many values of $\lambda$, which is the case as each pair of elements can swap at most once due to linearity in $\lambda$. To argue about these orderings, for any $\lambda$, we write $x_1(\lambda) \le x_2(\lambda) \le \dots \le x_{|\mathcal{X}|}(\lambda)$ for the sorted order of $\mathcal{X}$ under $\lambda$. Let $\rank_\lambda(x)$ of an $x\in \mathcal{X}$ denote the rank of $x$ in the sorted order of $\mathcal{X}$ -- note that since different elements in $\mathcal{X}$ might have the same value under $\lambda$, this value is not uniquely determined. Thus, we define $\rank_\lambda(x)$ as the interval $I$ of indices $i$ such that $x_i(\lambda) = x(\lambda)$.  Note that $\rank_\lambda(x)$ is always an interval in $\{1, \dots, |\mathcal{X}|\}$. We use the analogous notion of $\rank_\lambda(y)$ for any $y\in \mathcal{Y}$ to denote the interval of elements equal to $y(\lambda)$ in the sorted order of $\mathcal{Y}$ under $\lambda$.

A technical complication is that all $R_{i,j}(\lambda)$'s are \emph{open} rectangles, while we aim to reduce to the problem of maintaining a union of closed rectangles. We overcome this complication by replacing each rank value $r\in \{1, \dots, |\mathcal{X}|\}$ by two (symbolic) coordinates $\en(r), \st(r)$, with the understanding that $\en(r)$ is used for open intervals ending at $r$ and $\st(r)$ is used for open intervals starting in $x$. For intuition, one may think of $\en(r)$ and $\st(r)$ as $r-\eps$ and $r+\eps$, respectively. The corresponding set
\[
\mathcal{A} \coloneqq \bigcup_{r \in [|\mathcal{X}|]} \{\en(r), \st(r)\}
\]
can be viewed as $\{1, \dots, 2|\mathcal{X}|\}$, where $\en(r) \coloneqq 2r - 1$ and $\st(x) \coloneqq 2r$. Thus, we have the ordering $\en(1)<\st(1)<\en(2) < \st(2) < \cdots < \en(|\mathcal{X}|) < \st(|\mathcal{X}|)$. We define $\mathcal{B} \coloneqq \bigcup_{r \in [|\mathcal{Y}|]} \{\en(r), \st(r)\}$ analogously.

We can now define a representation of $R_{i,j}(\lambda)$ in rank space as follows: For any  $R_{i,j}(\lambda) = (a_{i,j}(\lambda), b_{i,j}(\lambda))\times (c_{i,j}(\lambda), d_{i,j}(\lambda))$, we define its \emph{closed rank representation}  $C_{i,j}(\lambda)$ as 
\begin{align*}
	C_{i,j}(\lambda) \coloneqq & [\st(\max \rank_\lambda(a_{i,j}(\lambda))), \en(\min \rank_\lambda(b_{i,j}(\lambda)))] \times \\
	& [\st(\max \rank_\lambda(c_{i,j}(\lambda))), \en(\min \rank_\lambda(d_{i,j}(\lambda)))],
\end{align*}
where $\min$ and $\max$ are taken over intervals and evaluate to the start and end, respectively.

The following lemma proves equivalence of the two representations.
\begin{lem}\label{lem:rank-equivalence}
	Let $x,x'\in \mathcal{X}, y,y' \in \mathcal{Y}$, as well as $r = \min \rank_\lambda(x), r'=\max \rank_\lambda(x')$ and $s= \min \rank_\lambda(y), s' = \max \rank_\lambda(y')$. Then,
	\[ \bigcup_{i,j} R_{i,j}(\lambda) \text{ covers } [x,x']\times[y,y'] \iff \bigcup_{i,j} C_{i,j}(\lambda) \text{ covers } [\en(r),\st(r')]\times[\en(s),\st(s')].\]
\end{lem}
\begin{proof}
	Let $x=x_1 < x_2 < \dots < x_R=x'$ denote the distinct values in $\mathcal{X}\cap [x,x']$ under $\lambda$, and likewise let $y=y_1 < y_2 < \dots < y_S=y'$ denote the distinct values in $\mathcal{Y}\cap [y,y']$ under $\lambda$. We partition $[x,x']\times[y,y']$ into regions in $\mathcal{I}\times \mathcal{J}$, where 
	\begin{align*} 
		\mathcal{I} = \{ \{x_1\}, (x_1, x_2), \{x_2\}, \dots, (x_{R-1},x_R), \{x_R\}\}\\  
		\mathcal{J} = \{ \{y_1\}, (y_1, y_2), \{y_2\}, \dots, (y_{S-1},y_S), \{y_S\}\}  
	\end{align*}
	To each interval of the form $\{x_i\}$, we associate the interval
\[
	\{\en(\min \rank_\lambda(x_i)), \st(\max \rank_\lambda(x_i))\},
\]
and to each interval of the form $(x_i, x_{i+1})$, we associate the interval
\[
\{\st(\max \rank_\lambda(x_i)), \en(\min \rank_\lambda(x_{i+1}))\},
\]
with the analogous associations for intervals in $\mathcal{J}$.
The claim is that each $\tau \in I\times J$ for $I\in \mathcal{I}, J\in \mathcal{J}$ is covered by $\bigcup_{i,j} R_{i,j}(\lambda)$ if and only if the Cartesian product of the associated intervals is covered by $\bigcup_{i,j} C_{i,j}(\lambda)$. We show the claim for the case that $I=\{x_k\}, J=(y_\ell,y_{\ell+1})$, all other cases are analogous.

	We have that 
	\begin{align*}
		& \tau \in \{x_k\}\times (y_\ell,y_{\ell+1}) \text{ is covered by } \bigcup_{i,j} R_{i,j}(\lambda) \\
		\iff & \exists i,j: R_{i,j}(\lambda)=(x_a, x_b) \times (y_c, y_d) \text{ with } x_a < x_k < x_b \text{ and } y_c \le y_\ell < y_{\ell+1} \le y_d\\
		\iff & \exists i,j: C_{i,j}(\lambda)=[\st(a'), \en(b')] \times [\st(c'), \en(d')] \text{ with } \\
		& \qquad \qquad a' < \min \rank_\lambda(x_k) \le \max \rank_\lambda(x_k) < b',\\
		& \qquad \qquad c' \le \min \rank_\lambda(y_\ell) < \max \rank_\lambda(y_{\ell+1}) \le d'\\
		\iff & \{\en(\min \rank_\lambda(x_k)), \st(\max \rank_\lambda(x_k))\} \\  
		&\times \{\st(\max \rank_\lambda(y_\ell)), \en(\min \rank_\lambda(y_{\ell+1}))\} \text{ is covered by } \bigcup_{i,j} C_{i,j}(\lambda).
	\end{align*}
	where the last equivalence is notable, as it uses that all $C_{i,j}(\lambda)$ are of the form $[\st(a'),\en(b')]\times [\st(c'),\en(d')]$.
\end{proof}

Thus, by checking whether the closed rectangles $C_{i,j}(\lambda)$ cover the area corresponding to the bounding box~$B$ of $Q$, we can determine whether $\lambda$ is a feasible placement. To reduce this further to the question whether $\bigcup_{i,j} C_{i,j}(\lambda)$ covers the full area $[1,|\mathcal{A}|]\times [1,|\mathcal{B}|]$, we introduce four additional rectangles $C_L, C_R, C_T, C_B$ covering everything \emph{but} the bounding box $B$:
\begin{align*}
	C_L & = [1, \en(\min \rank_\lambda(x_0^B))] \times [1,|\mathcal{B}|]\\
	C_R & = [\st(\max \rank_\lambda(x_1^B)), |\mathcal{A}|] \times [1,|\mathcal{B}|]\\
	C_B & = [1,|\mathcal{A}|] \times [1, \en(\min \rank_\lambda(y_0^B))]\\
	C_T & = [1,|\mathcal{A}|] \times [\st(\max \rank_\lambda(y_1^B)), |\mathcal{B}|]
\end{align*}

Let $\mathcal{C}_\lambda = (\bigcup_{i,j} C_{i,j}(\lambda)) \cup \{C_L, C_R, C_B, C_T\}$ denote the set of rectangles in our closed rank representation.

The main idea of our algorithm is to start from some value $\lambda_0$ that trivially satisfies $\lambda_0 > \lambda^*$, and decrease $\lambda$ while maintaining a data structure for $\mathcal{C}_\lambda$ that after any change allows us to check whether $\mathcal{C}_\lambda$ still covers the full space $[1,|\mathcal{A}|]\times [1,|\mathcal{B}|]$. Thus, when does $\mathcal{C}_\lambda$ change while $\lambda$ decreases?

Clearly, $\mathcal{C}_\lambda$ only changes when the sorted order of $\mathcal{X}$ or $\mathcal{Y}$ changes, i.e., at an intersection point $\lambda$ where $x(\lambda) = x'(\lambda)$ for some $x,x'\in \mathcal{X}$ or $y(\lambda)=y'(\lambda)$ for some $y,y'\in \mathcal{Y}$. Let $\lambda_1 \ge  \dots \ge  \lambda_L$ denote the descendingly sorted order of such $\lambda$, which we call \emph{critical values}. Since every $x\in \mathcal{X}$ and every $y\in \mathcal{Y}$ is a linear (or constant) function in $\lambda$, we obtain that there are at most $\Oh(|\mathcal{X}|^2+|\mathcal{Y}|^2) = \Oh((pq)^2)$ such values. Note that we can determine $\lambda_1, \dots, \lambda_L$ in time $\Oh((pq)^2\log pq)$. 

The critical values $\lambda_1, \dots, \lambda_L$ partition $\mathbb{R}$ into the regions
\[
(\infty, \lambda_1), \{\lambda_1\}, (\lambda_1,\lambda_2), \dots, \{\lambda_L\}, (\lambda_L, -\infty),
\]
where $C_\lambda$ is uniform (unchanged) in each $(\lambda_i, \lambda_{i+1})$. As is intuitively clear, $\lambda^*$ can only be one of the critical values, formally proved below.

\begin{prop}\label{prop:criticaloptimality}
We must have $\lambda^*=\lambda_i$ for some $i\in[L]$.
\end{prop}
\begin{proof}
	Consider a feasible translation $\tau$ for some scaling factor $\lambda \in (\lambda_i, \lambda_{i+1})$. Note that $\tau \in [x_a,x_{a+1}]\times [y_b,y_{b+1}]$ for some indices $a,b$ in the sorted orders of $\mathcal{X},\mathcal{Y}$ under $\lambda$. Observe that any $\tau'\in [x_a,x_{a+1}]\times[y_b,y_{b+1}]$ remains a feasible translation for every $\lambda \in [\lambda_i,\lambda_{i+1})$, in particular for $\lambda = \lambda_i$: none of the open rectangles $R_{i,j}$ can cover any point in $[x_a,x_{a+1}]\times[y_b,y_{b+1}]$, and this region remains non-empty (but possibly consisting of a single point) for $\lambda = \lambda_i$.
\end{proof}

We can finally give the main outline of our algorithm, see Algorithm~\ref{alg:main}.

\begin{algorithm}
\begin{algorithmic}[1]
\Function{LargestCopy}{$P$, $Q$}
	\State Compute critical values $\lambda_1 \ge \cdots \ge \lambda_L$
	\State Compute $\mathcal{C} \gets \mathcal{C}_\lambda$ for $\lambda \in (\infty,\lambda_1)$.
	\For{$i\gets 1$ \textbf{to} $L$} 
		\State Update $\mathcal{C}$ to $\mathcal{C}_\lambda$ for $\lambda = \lambda_i$ \label{line:update-1}
		\If{$\mathcal{C}$ does not fully cover $[1,|\mathcal{A}|]\times [1,|\mathcal{B}|]$}
			\State \Return $\lambda_i$
	        \EndIf	
		\State Update $\mathcal{C}$ to $\mathcal{C}_\lambda$ for $\lambda \in (\lambda_i, \lambda_{i+1})$ \label{line:update-2}
	\EndFor
\EndFunction
\end{algorithmic}
	\caption{Computing $\lambda^*$, the largest scaling factor $\lambda$ such that $\lambda P$ can be translated to be contained in $Q$}
\label{alg:main}
\end{algorithm}

We now show that Algorithm \ref{alg:main} fulfills the requirements of Lemma \ref{lemma:reduction}.

\begin{proof}[Proof of Lemma \ref{lemma:reduction}]
	By Proposition~\ref{prop:criticaloptimality} and Lemma~\ref{lem:rank-equivalence}, correctness of Algorithm \ref{alg:main} is immediate. It remains to argue how we can implement this algorithm to achieve the desired upper bound.
	As noted above, computing the sorted critical values $\lambda_1, \dots, \lambda_L$ takes time $\Oh((pq)^2 \log pq)$. Afterwards, we need to maintain a data structure for $\mathcal{C}$ such that we can answer whether $\mathcal{C}$ fully covers $[1,|\mathcal{A}|]\times [1,|\mathcal{B}|]$.  Observe that any changes to $C_{i,j}(\lambda)$ or $C_L, C_R, C_B, C_H$ occur at some critical value $\lambda_\ell$. Specifically, a coordinate of $C_{i,j}(\lambda)$ changes only if $a_{i,j}(\lambda)$ or $b_{i,j}(\lambda)$ (respectively, $c_{i,j}(\lambda)$ or $d_{i,j}(\lambda)$) is equal to some other $x(\lambda)$ for some $x\in \mathcal{X}$ (respectively $y(\lambda)$ for some $y\in \mathcal{Y}$) --- note that the corresponding value will undergo at most two changes, one at $\lambda = \lambda_i$ (where it is equal to other coordinates) and one for $(\lambda_i, \lambda_{i+1})$ (between critical values, all coordinates are indeed distinct). Similarly, any change to $C_L, C_R, C_B, C_T$ only occurs if under $\lambda$ there is equality of $x_0^B$ or $x_1^B$ with some other $x\in \mathcal{X}$, or of $y_0^B$ or $y_1^B$ with some other $y\in \mathcal{Y}$. By charging each update to a pair $(x,x') \in \mathcal{X}^2$ or $(y,y')\in \mathcal{Y}^2$ responsible for the update (resulting from $x(\lambda)=x'(\lambda)$ or $y(\lambda) = y'(\lambda)$), we can bound the total number of updates to $\mathcal{C}$ by $\Oh(|\mathcal{X}|^2+|\mathcal{Y}|^2)=\Oh((pq)^2)$, as each pair $(x,x')$ or $(y,y')$ is charged at most a constant number of times and either satisfies $x(\lambda)=x'(\lambda)$ or $y(\lambda)=y'(\lambda)$ for all $\lambda$ (and thus is responsible for no updates) or has a unique point of equality (since each $x\in \mathcal{X}$ and $y\in \mathcal{Y}$ is a linear function in $\lambda$).

In total, we perform $U=\Oh((pq)^2)$ many updates to $\mathcal{C}$. Observe that we can precompute these updates in such a way that we search for the first update after which the collection $\mathcal{C}$ of rectangles no longer covers the full region $[1,|\mathcal{A}|]\times [1,|\mathcal{B}|]$.\footnote{The very observant reader might notice that there is a slight technicality: at each update in Lines~\ref{line:update-1} and~\ref{line:update-2}, we should first add all new boxes before we delete old boxes in order to not create artificial holes while updating. Note however, that this results in having, at any given time, only at most twice as many boxes in our data structure compared to the initial promise of $p'q' + 4$ boxes.} Thus, using a data structure for \textsc{Offline Dynamic Rectangle Cover}, we can simulate Algorithm~\ref{alg:main} in time $\Oh((pq)^2 \log(pq) + T(pq, (pq)^2))$.
\end{proof}

\bibliographystyle{alpha}
\bibliography{biblio}


\end{document}